\documentclass[letterpaper,12pt]{article}

\usepackage{amsmath}
\usepackage{amsthm}
\usepackage{amssymb}
\usepackage[]{graphicx}
\usepackage{fancyhdr}
\usepackage{wrapfig}
\usepackage{fullpage}
\usepackage{comment}
\usepackage{natbib}
\usepackage[dvipsnames]{xcolor}
\usepackage{mathrsfs}
\usepackage[multiple]{footmisc}
\usepackage{subcaption}
\usepackage{array}
\usepackage{multirow}
\usepackage{tabularx}
\usepackage{xurl}
\usepackage[hidelinks]{hyperref}
\usepackage{booktabs}
\usepackage{graphicx}
\usepackage{float}

\usepackage{tikz}
\usetikzlibrary{decorations.pathreplacing, arrows.meta}
\tikzset{vert/.style = {circle, fill, inner sep = 0, minimum size = 5}}
\newcommand{\circled}[2][inner sep=1pt]{\ifmmode
\tikz[baseline=(X.base),outer sep=0pt]{\node[circle,draw,#1](X){\ensuremath{#2}};}
\else
\tikz[baseline=(X.base),outer sep=0pt]{\node[circle,draw,#1](X){#2};}
\fi
}

\newenvironment{tabnotes}[2][1]{\begin{minipage}[t]{#1\textwidth}\vspace{0.1cm}\scriptsize{\emph{Notes:} #2}}{\end{minipage}}

\newtheorem{thm}{Theorem}
\newtheorem{prop}{Proposition}
\newtheorem{lm}{Lemma}

\theoremstyle{definition}

\newtheorem{hypothesis}{Hypothesis}

\title{Priority Transparency, Admission Chances, and Information Acquisition in School Choice}
\author{Georgy Artemov\footnote{Department of Economics, University of Melbourne, VIC 3010, Australia, gartemov@unimelb.edu.au.} \and Siqi Pan\footnote{Department of Economics, University of Melbourne, VIC 3010, Australia, siqi.pan@unimelb.edu.au.}}
\date{\vspace{-3ex}}

\usepackage{version}
\usepackage[normalem]{ulem}

\begin{document}
\maketitle

\begin{abstract}

We study, theoretically and experimentally, how transparency about students' priorities and admission chances shapes their incentives to acquire information about their own preferences in school choice and college admissions. In the model, uninformed students choose schools based on a common prior. When they learn their own preferences, their choices become more heterogeneous, which frees up seats at popular schools. Students who know they have high priority have stronger incentives to learn because they can more readily act on what they learn, whereas students who know they have low priority are discouraged. Full priority disclosure concentrates learning among high-priority students. By pooling priorities, partial disclosure spreads learning incentives to pooled students and yields higher welfare. In the laboratory, however, full disclosure yields the highest welfare instead, followed by partial disclosure, and then no disclosure, because greater transparency improves subjects' understanding of the strategic environment, leading to fewer mistakes. These findings support full disclosure of priorities or admission chances to guide information acquisition. However, deviations in learning remain even under greater priority transparency, partly because subjects respond suboptimally to admission chances when these are provided directly rather than inferred. Students' ability to interpret and use them is therefore itself a policy concern.

%We study, theoretically and experimentally, how priority transparency and admission chances shape students’ incentives to acquire information about their own preferences in school choice. In the model, learning one’s preferences reduces congestion at popular schools. Under full disclosure, students who know they have low priority are discouraged from learning, whereas partial disclosure pools priorities and may preserve lower-priority students’ incentives, yielding higher welfare. In the laboratory, however, full disclosure performs best because it improves subjects’ understanding of the strategic environment and reduces mistakes. Deviations remain even when admission chances are stated rather than inferred, making students’ use of such information a policy concern.

\end{abstract}

{\small

}

\section{Introduction}

Market design has seen remarkable success in organizing real-world allocation processes, including centralized assignment systems for school choice and college admissions. Yet as these mechanisms have moved from theory into practice, new considerations have emerged. Among them are participants' demands for greater transparency about how outcomes are determined and their chances of receiving different assignments, as well as their imperfect information about their own preferences and the costly learning required to refine it. This paper explores how transparency shapes learning. We show that limiting transparency can improve welfare in theory, but make the resulting incentives harder for participants to respond to in practice.

An important dimension of transparency is how precisely students know their admission priorities and how readily they can translate them into admission chances. Access to such knowledge varies widely across the world and is actively evolving within many systems. In New York City, parents successfully campaigned for the release of lottery numbers used to break ties in admission priorities; the Department of Education now publishes these numbers alongside a high/medium/low admission-chance indicator for each program.\footnote{New York City's middle- and high-school choice systems typically use a single lottery to break ties within priority groups and, at unscreened schools, to determine priorities entirely. Lottery numbers were historically withheld. Following demands from PLACE NYC and several elected parent councils, the Department of Education began displaying them in MySchools in February 2022. See \cite{marian2023algorithmic} for details.} Similar moves toward greater transparency have occurred in systems with exam-based priorities. China, for example, shifted from requiring students to submit university rank-order lists before learning their exam results to releasing exam results before submission; Queensland, Australia, moved from reporting exam outcomes in broad bands to reporting precise ranks; and Chile introduced a tool that estimates applicants' admission chances.\footnote{By 2017, virtually all provinces in China had completed the transition to rank-order list submission after the release of exam results \citep{lien2016preference,lien2017ex,pan2019instability}. Until 2019, Queensland reported tertiary-entrance standing as an ``Overall Position'' (OP), a coarse 1--25 band; from 2020, it adopted the Australian Tertiary Admission Rank (ATAR), reported on a 0--99.95 scale in 0.05 increments. For details about Chile's online admission simulator, see \url{https://demre.cl/noticias/2016-11-17-simulador-postulaciones}.} 

Even when centralized systems leave participants to infer their own chances, third-party tools may emerge to meet that demand. In Indian medical, dental, and engineering admissions, applicants know their exam results but cannot readily infer which colleges those results are likely to secure. An industry of ``college predictors'' combines previous years' admission cutoffs with current-year seat capacities to estimate applicants' chances across their options.\footnote{See \cite{baswana2019}. Examples of ``college predictors'' include \url{https://allen.in/neet/college-predictor}, \url{https://collegedunia.com/neet-college-predictor}, and \url{https://engineering.careers360.com/jee-college-predictor}.}

Sometimes, reforms driven by other objectives can inadvertently reduce transparency. When school districts---including Toronto's specialized programs---replaced selective admissions with lotteries to promote equity, they also made applicants' chances less predictable, as lottery numbers were not disclosed.\footnote{In 2022, the Toronto District School Board replaced auditions, portfolios, and report-card assessments for specialized arts, athletics, and STEM programs with a lottery, effective September 2023. See \nolinkurl{https://www.cbc.ca/news/canada/toronto/toronto-district-school-board-specialized-schools-programs-admissions-1.6466071}.} Importantly, a shift to lotteries need not reduce predictability---as the NYC example shows, disclosing lottery numbers can restore it.

In many of these examples, public discussions connect transparency about admission chances to students' decisions about how much to invest in learning about their options. Learning which schools are the best fit can require substantial time and resources: parents and students may need to visit schools, attend information sessions, talk to current students, or investigate course offerings. In New York City, one argument for releasing lottery numbers was that they would help families plan their school search \citep{marian2023algorithmic}. Similarly, in India and Chile, counseling services and universities advise students to take examination results into account when researching their options.%
\footnote{For India, see \url{https://www.doctutorials.com/mbbs-curriculum/neet-ug-counselling-college-selection} and \url{https://www.pw.live/neet/exams/mistakes-to-avoid-while-applying-for-the-neet-ug-counselling}; for Chile, see \url{https://portaleduca.cl/entrega-resultados-paes-formando-chile-entrega-5-recomendaciones-para-postular-sin-errores/} and \url{https://uchile.cl/noticias/140179/antes-de-postular-a-la-universidad-que-hay-que-considerar}.}

These discussions reflect a basic economic intuition: the value of acquiring information depends on students' priorities and admission chances. A student with higher priority, and therefore higher admission chances, is more likely to be able to act on what she learns and has a stronger incentive to do so; a student with lower priority has weaker incentives for the same reason. Despite the importance of this link, its policy implications have, to our knowledge, received little attention. This paper formalizes the link and studies how priority transparency shapes students' incentives to learn about their own preferences. Our results show that failing to account for it can carry a real cost.

We develop a school choice model with costly information acquisition. Students share a common prior over schools, capturing public rankings or widely shared perceptions of school quality. By paying a cost, they can learn about their idiosyncratic preferences, reflecting student-school fit. When students rely mainly on common rankings, demand becomes concentrated on the same popular schools. Learning can make choices more heterogeneous: some students may discover that a less popular school is a better fit.\footnote{\citet{Narita(2018)} studies New York City high-school choice and finds that, after search, students revise their choices to become more responsive to school characteristics related to their personal circumstances, including distance, language, and current students.} This can reduce congestion and increase other students' admission chances. Learning therefore generates a positive externality that students do not fully internalize.

Such a model captures a concern prominent among education experts: students rely on aggregate rankings rather than researching their options. Richard Beeman (\citeyear{beeman_learning_beyond_measure_2002}), dean of undergraduate education at the University of Pennsylvania, argued that ``the very idea that\dots rankings can be useful to students, is highly problematic,'' while Frank Bruni (\citeyear{bruni_college_rankings_2023}), a Duke professor and author of a NYT bestseller on college admissions, observed that students ``couldn't really tell me why they wanted Duke---they just knew that they should want it [because it] was highly ranked.'' This pattern is exactly what happens in our model at the individual level, but we take this further to explore how aggregate welfare is affected by congestion-reducing externalities and transparency about admission chances.

More precise knowledge of admission chances affects welfare not only through how many students learn, but also through who learns. Students with strong admission chances are more likely to benefit from learning about their idiosyncratic fit and thus have greater incentives. Those with low admission chances, by contrast, may be discouraged from learning altogether, even if doing so would have been valuable to them. Our theoretical results show that, although precise knowledge is privately valuable, full transparency may not be welfare-maximizing. Coarser disclosure can pool students with different priorities and preserve stronger learning incentives among students who would otherwise stop learning.

While theory suggests that the optimal information structure can be coarser than full disclosure, a coarser structure may also impose greater cognitive demands on participants. Its practical performance therefore depends on whether participants understand it and respond correctly. Recent work in market and mechanism design shows that agents may fail to respond correctly even when their incentives are theoretically simple in theory but hard to recognize in the mechanism's description.\footnote{Deviations from dominant strategies in strategy-proof mechanisms are well documented \citep{hassidim2017mechanism,hassidim2017redesigning,artemov2017strategic,reesjones2018suboptimal,reesjones2018experimental,hassidim2021limits,shorrer2024dominated}. One explanation is that participants may not recognize the incentives created by a mechanism or its standard description. \citet{li2017obviously} formalizes this idea through obvious strategy-proofness, and \citet{li2024designing} surveys related approaches to mechanism simplicity. Recent experiments by \citet{katuscak2024strategy} and \citet{gonczarowski2024describing} study whether descriptions that make strategy-proofness more apparent improve understanding and truthful reporting. For surveys of behavioral responses in market-design experiments, see \citet{hakimov2020experiments} and \citet{pan2020experiments}.} We therefore conduct a laboratory experiment comparing policies that provide students with no, partial, or full knowledge of their priorities.

Our experimental results reveal a discrepancy with the theoretical prediction: more precise knowledge about priorities tends to yield higher social welfare. Specifically, full knowledge dominates partial knowledge, which in turn outperforms the no-knowledge policy. Individual-level analysis suggests that the reason is behavioral. Subjects in the no-knowledge treatment have greater difficulty understanding the strategic environment and mapping others’ learning decisions into their own admission chances. Greater transparency reduces this complexity and leads to fewer strategic mistakes. However, difficulties in learning decisions are not purely strategic: even with greater transparency, deviations from theoretically optimal learning decisions remain and are partly explained by suboptimal responses to directly provided admission chances.

Our results provide a clear policy implication: full priority transparency should be provided early enough to guide students' information acquisition, subject to constraints on the admissions process. Indeed, our experimental findings may understate the value of transparency in real admission systems. In our environment, students either learn their entire preference realization or do not learn at all, and the mechanism does not require them to engage in complex strategizing when submitting rank-order lists. Yet even here, full disclosure outperforms the theoretically superior partial disclosure, in terms of welfare, by reducing behavioral mistakes. In practice, transparency may also help students target their search toward relevant schools and submit better rank-order lists in non-strategy-proof mechanisms, such as immediate acceptance or deferred acceptance with constrained lists. These additional channels lie outside our model and suggest that the value of priority transparency may exceed what our experimental environment is designed to capture.

Lastly, our experiment shows that directly providing admission chances does not by itself ensure optimal information-acquisition decisions. This is an important policy concern. When New York City Public Schools introduced a prediction tool in 2024 that classified each applicant's chance of receiving an offer from a high-school program as high, medium, or low, families welcomed the effort to provide more guidance, yet some raised concerns that these coarse categories are confusing and may undermine trust in the predictions.\footnote{See \url{https://medium.com/algorithms-in-the-wild/nyc-high-school-chances-of-admission-predictions-cb15fd4b5655}.} Our findings suggest that such concerns are justified and that authorities need to think carefully about how admission-chance information is presented, whether students understand it, and what decision support is needed to help them act on it.

\subsection{Related Literature}

Recent work on information acquisition in matching markets relaxes the standard assumption that participants know their own preferences and can rank all available options at no cost. Instead, it recognizes that participants in many real-world matching markets need to engage in costly search to learn about their own preferences.\footnote{Empirical evidence shows that students actively search and learn about their preferences in school choice and college admissions. \citet{Narita(2018)} finds that a substantial share of students in New York City revise their choices after acquiring additional information about schools. \citet{grenet2022preference} find that university programs from which students receive earlier offers are more likely to be ranked higher, consistent with students engaging in costly search over their options.}

Within this literature, the strand most closely related to our paper studies the externalities and social optimality of information acquisition in matching environments. In these models, agents share a common prior over options, which generates congestion at popular schools; information acquisition can relieve this congestion by making choices more heterogeneous and thereby generate positive externalities. \citet{artemov2021assignment} shows that privately optimal information acquisition may fall below the social optimum because students do not internalize how their search affects other students. He argues that fully disclosing priorities would improve welfare---a policy that, in our model, may not be optimal. \citet{maxey2024school} compares simultaneous and sequential search protocols under different degrees of preference correlation and identifies a policy instrument through which the designer can coordinate search and improve welfare. In the continuum model of \citet{maxey2024school}, there is no aggregate uncertainty, so every student perfectly predicts her assignment. As a result, there are no marginal students whose admission chances depend on cutoff realization, leaving no room for intermediate disclosure policies. Our paper extends this strand by allowing for partial disclosure and bringing the theory to the laboratory. We show that behavioral responses to an information structure that is theoretically optimal but more complex than full disclosure can overturn the model's welfare ranking.

While the papers above study information acquisition within a fixed matching mechanism, another strand studies how the choice of matching mechanism shapes participants' incentives to acquire information. \citet{bade2015} introduces information acquisition into a house-allocation problem in which agents must incur a cost to learn their own preferences. She shows that serial dictatorship provides correct learning incentives and is the unique ex-ante Pareto-optimal, strategy-proof, and nonbossy mechanism. \citet{harless2018learning} also study strategic learning in object allocation and show that top trading cycles can dominate serial dictatorship under inequality-averse welfare criteria. In school choice, \citet{chen2022information} show that only immediate acceptance provides incentives for students to learn both their own cardinal preferences and the preferences of other participants, and that immediate acceptance and deferred acceptance cannot be unambiguously ranked in terms of welfare. \citet{chen2021information} test these predictions experimentally and find that students are willing to pay significantly more for information under immediate acceptance than under deferred acceptance, with an overall tendency toward over-searching across treatments. Unlike these papers, we hold the matching mechanism fixed and study how priority transparency affects students' incentives to acquire information.

Another related strand studies interventions that provide direct information about admission chances or feasible options, focusing on how such information affects students' decisions about which schools to search. \citet{arteaga2022smart} study personalized feedback about admission chances in centralized school choice and show that such feedback can improve search and application decisions. \citet{immorlicaetal_2018} show that students benefit from learning their feasible sets before acquiring costly information about universities, but that this can generate information deadlocks when students delay search while waiting for others' decisions. \citet{hakimov2023costly} compare direct and sequential serial dictatorship and study the provision of cutoff information. They find that sequential serial dictatorship and cutoff provision can improve welfare by directing students' information acquisition toward schools that remain feasible, and that the sequential mechanism reduces deviations from optimal search behavior. \citet{noda2022strategic} studies the optimal disclosure of feasible sets under random serial dictatorship and shows that full disclosure can be Pareto inefficient because of positive externalities in information acquisition. Unlike these papers, our focus is not on how direct information about admission chances or feasible options affects which schools students search, but on how disclosed priorities shape students' beliefs about their admission chances, and how those beliefs affect overall learning incentives and social welfare. \footnote{Another strand of the literature documents that applicants facing lower admission chances are more likely to submit dominated or otherwise nontruthful preference reports \citep{hassidim2017mechanism,hassidim2017redesigning,artemov2017strategic,reesjones2018suboptimal,reesjones2018experimental,hassidim2021limits,shorrer2024dominated}. By contrast, in our model, choosing not to acquire information is a rational response to one's admission chances rather than a mistake in preference reporting.}

Whereas our paper studies how students' knowledge of their priorities or admission chances affects information acquisition, other work focuses on the preference-submission stage. In settings where priorities are determined by entrance-exam performance, \citet{lien2016preference} and \citet{lien2017ex} study whether exam outcomes should be disclosed before or after students submit rank-order lists. Under unbiased beliefs about exam performance, they show that withholding exam outcomes until after preference submission can generate ex-ante fair and efficient outcomes under immediate acceptance. \citet{pan2019instability} experimentally studies this question when students form beliefs about their performance in a real-effort task. Biased self-assessments weaken the case for delayed disclosure under immediate acceptance, whereas disclosure before preference submission under deferred acceptance improves stability and yields a more equitable distribution of welfare. \citet{huang2025transparent} study whether lottery outcomes used to break priority ties should be revealed before students submit rank-order lists, showing that disclosure can improve strategic reporting and welfare when lists are constrained. \citet{haeringer2026information} examine how different framings of admission-chance information affect preference submission and find large and heterogeneous framing effects.

Priority transparency is part of a broader set of questions about transparency, verifiability, and auditability in matching mechanisms. \citet{hakimov2025transparency} study how feedback, such as cutoffs or information generated by single-school elicitation, can allow students to verify their assignments and the implementation of the announced admissions procedure without disclosing other students' preferences or priorities. \citet{moller2026transparent} studies the trade-offs among transparency, strategy-proofness, stability, and efficiency in standard one-to-one matching markets with limited commitment.
\citet{grigoryan2024auditability} compare school-choice rules according to how much information is required to detect deviations from their prescribed implementation. In these papers, transparency concerns whether an allocation rule can be credibly implemented or audited. Here, by contrast, transparency refers to the disclosure of a student's own priority and admission chances.

\section{Experimental Setup and Theoretical Analysis}
\label{sec:theory}

In this section, we introduce the environment and derive theoretical predictions. While we focus on the experimental setup here, in \ref{sec:general_model}, we replicate key insights in a continuum model.

\subsection{Setup}
In each experimental market, there are three students and two schools, A and B. School A has one available seat, while School B has enough seats for all students. Each student is assigned to one of the two schools. A student admitted to School A receives a payment of 30 AUD. A student admitted to School B receives a payment of either 40 AUD or 10 AUD, with equal probability, and independently across students. Experimental subjects act as students; schools are not strategic, and their actions are simulated by the computer.

Students first decide whether to acquire information about their own preferences, that is, whether to learn their realized payoff from School B: 40 AUD or 10 AUD. Acquiring this information incurs a cost and allows the student to observe the realized School-B payoff before the admission procedure. 

Once students have made their learning decisions, they are assigned to schools through a sequential Random Serial Dictatorship (RSD) mechanism.\footnote{Because students have the same priority ordering at both schools, the Deferred Acceptance and Top Trading Cycles mechanisms reduce to the Serial Dictatorship mechanism. Our paper primarily focuses on learning decisions rather than the decisions to choose between schools or submit preferences. Therefore, we intentionally choose a sequential mechanism to simplify subjects' decisions in the admission procedure after information acquisition. Several studies have demonstrated the behavioral advantages of sequential over static matching mechanisms (see, for example, \citealp{li2017obviously}; \citealp{Pycia2019}; \citealp{klijn2019static}; \citealp{bo2020iterative}). } Students' priorities in RSD are randomly determined, with each student having an equal chance of being ranked first (``Student 1''), second (``Student 2''), or third (``Student 3''). As described in the next section, depending on the treatment, students may have no, partial, or full knowledge about their priorities before deciding whether to acquire information. 

In the admission procedure, students make school decisions one by one in order of their priorities:
\begin{description}
\item [{Step 1:}] Student 1 chooses between Schools A and B;

\item [{Step 2:}] If School A is still available, Student 2 chooses between A and B. Otherwise, she is automatically admitted to B.

\item [{Step 3:}] If School A is still available, Student 3 chooses between A and B. Otherwise, she is automatically admitted to B.
\end{description}
A student's final payoff equals the payoff from the school to which she is admitted, minus the learning cost if she chooses to learn.

\subsection{Treatments}
The treatments differ in what students know about their priorities before deciding whether to acquire information. In Treatment \textit{Unknown}, students do not know their priorities when making their learning decisions. Their priorities are revealed only after the information-acquisition stage and before the admission procedure. We refer to all students in this treatment as ``Student 1/2/3.'' In Treatment \textit{Known}, students know their exact priorities when deciding whether to learn about School B and are referred to directly by those priorities. In Treatment \textit{Coarse}, students are told only whether they are among the top two or have the lowest priority when making their learning decisions. Students in the top-two group have the first or second priority with equal probability, whereas the remaining student always has the third priority.\footnote{For the cost combinations used in our main experiment, Experiment 2, this is the optimal disclosure policy because it implements the socially optimal learning pattern: Students 1 and 2 learn, while Student 3 does not. For other cost combinations, more complex disclosure policies may be optimal. For example, when Student 1's cost is above 3.75 and Student 2's cost is below 3.75, it may be optimal to pool Student 1 with only some Student 2 types.} The top-two students are labeled ``Student 1/2,'' and the student with the third priority is labeled ``Student 3.''

\subsection{Theoretical Analysis}
\label{sec:theoretical predictions}

We next conduct a theoretical analysis and derive equilibrium predictions. Although multiple equilibria arise in many settings, our analysis focuses on the maximum-information equilibrium because it delivers the highest total welfare. This is the equilibrium on which a designer would attempt to coordinate students, making it the most meaningful basis for welfare comparisons across treatments.\footnote{In other equilibria of \textit{Unknown} and \textit{Coarse}, students anticipate that others do not learn and do not learn themselves.} Equilibrium selection, though, should not be a concern: as discussed below, our experimental design allows us to fix students' beliefs about others' learning strategies and examine whether the observed welfare patterns are driven by equilibrium selection.

Under the payoff structure described above, School B has an expected payoff of 25, which is less than School A's guaranteed payoff of 30. Without knowing B's realized payoff, every risk-averse or risk-neutral student prefers School A. After learning, a student prefers A to B or B to A with equal probability. As more students acquire information, population preferences become more heterogeneous. Homogeneous preferences, driven by common prior beliefs, lead to congestion at School A; preference heterogeneity alleviates this congestion. Information acquisition therefore generates a positive externality by increasing other students' admission chances at School A, which in turn enhances their incentives to acquire information.

To illustrate this intuition, consider students' learning benefit under \textit{Known}, presented in Table \ref{tab:learning incentives}.\footnote{We assume risk neutrality in the model. As discussed in the experimental results, risk aversion may reduce subjects' incentives to acquire information, but it does not change the ordering of learning incentives across treatments or priority positions.} Student 1, with the highest priority, can always choose between Schools A and B and therefore has an admission chance of 100\% at School A. Her benefit from learning is 5, because learning increases her payoff by 10 with probability 50\%.
Student 2's learning benefit depends on Student 1's learning decision. If Student 1 does not learn, she chooses School A, leaving Students 2 and 3 with a 0\% chance of admission to School A. Not being able to choose between the two schools, neither Student 2 nor Student 3 has any reason to learn. If Student 1 learns, however, there is a 50\% chance that she discovers a preference for School B and chooses it instead. Student 2 then has a 50\% admission chance and a learning benefit of 2.5. If Student 2 does not learn, Student 3 has no incentive to learn. If Student 2 learns, however, Student 3 has a 25\% admission chance and a learning benefit of 1.25.

The analysis above shows that a student's admission chance is affected only by the learning decisions of students with higher priority. For a given student, we say another student is \textit{relevant} if there exists a priority realization in which that other student has higher priority. Under \textit{Unknown}, every student is relevant to every other student. Under \textit{Known}, no student is relevant to Student 1; Student 1 is relevant to Student 2; and Students 1 and 2 are relevant to Student 3. Under \textit{Coarse}, the other Student 1/2 is relevant to a Student 1/2, while both Students 1/2 are relevant to Student 3.

\begin{table}[htbp]
\centering
\caption{Learning Benefits and Admission Chances}
\label{tab:learning incentives}
\begin{tabular}{m{2.2cm}|m{2cm}|m{2.2cm}|m{2.2cm}|m{5.6cm}}
\hline
\textbf{Treatment} & \textbf{Student} & \textbf{Learning benefit} & \textbf{Admission chance} & \textbf{Conditions}\\
\hline
\multirow{3}{*}{\textit{Unknown}} & S1/2/3 & 1.67 & 33.3\% & No other students learn\\
& S1/2/3 & 2.08 & 41.7\% & One other student learns\\
& S1/2/3 & 2.92 & 58.3\% & Both other students learn\\
\hline
\multirow{5}{*}{\textit{Known}} & S1 & 5 & 100\% & None\\
 & S2 & 2.5 & 50\% & S1 learns \\
 & S2 & 0 & 0\% & S1 does not learn\\
 & S3 & 1.25 & 25\% & S1 and S2 both learn\\
 & S3 & 0 & 0\% & S1 and/or S2 does not learn\\
\hline
\multirow{5}{*}{\textit{Coarse}} & S1/2 & 3.75 & 75\% & Other S1/2 learns\\
& S1/2 & 2.5 & 50\% & Other S1/2 does not learn\\
& S3 & 1.25 & 25\% & Both S1/2 learn\\
& S3 & 0 & 0\% & At least one S1/2 does not learn\\
\hline
\end{tabular}
\begin{tabnotes}
    The table reports learning benefits---the payoff gain from acquiring information, excluding the learning cost---and admission chances across treatments and student priorities. These quantities depend on other students' learning decisions, as specified in the ``Conditions'' column. ``S'' followed by a number refers to a student with the corresponding priority.
\end{tabnotes}
\end{table}

Information acquisition exhibits strategic complementarity: a student's incentive to acquire information increases when she expects more students relevant to her to do so. Likewise, when students to whom she is relevant expect her to acquire information, their incentives to acquire information also increase. The resulting equilibrium chain reaction resembles a domino effect moving down the priority order. This helps explain why knowledge about priorities may enhance welfare. Students with higher priority, once aware of their stronger admission chances, are more likely to acquire information. In equilibrium, this raises lower-priority students' expected admission chances and can induce further information acquisition down the priority order.

At the same time, students with lower priority, once informed of their weaker admission chances, may choose not to acquire information. The designer therefore faces a trade-off: more precise knowledge about priorities can initiate the domino effect by encouraging higher-priority students to learn, but it can also weaken lower-priority students' incentives and halt the sequence prematurely. 

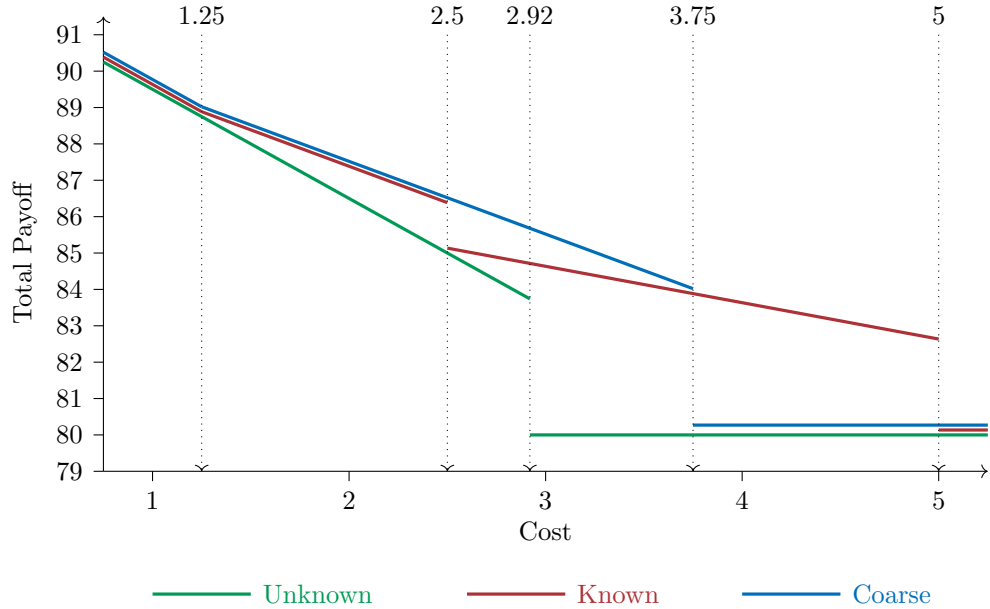
\begin{figure}[htbp]
\centering
\begin{tikzpicture}[scale=0.65]

    \def\xSc{4}
    \def\ySc{2}
    \def\yBase{80}      % payoff 80 is plotted at y=0
    \def\yCompress{2.7}   % vertical compression
    \def\tickSz{0.2}
    \def\xMin{0.75}
    \def\xMax{5.25}

    % Vertical coordinate of the x-axis, which represents payoff 79
    \pgfmathsetmacro{\xAxisY}{(79-\yBase)/\yCompress*\ySc}

    % Vertical coordinate of the lowest graph value, payoff 80
    \pgfmathsetmacro{\bWlf}{(80-\yBase)/\yCompress*\ySc}

    % x-axis
    \draw[->]
        (\xMin*\xSc,\xAxisY)
        --
        (\xMax*\xSc,\xAxisY);

    % x-axis ticks: correctly placed on the x-axis
    \foreach \i in {1,...,5}
        \draw
            (\i*\xSc,\xAxisY)
            --
            (\i*\xSc,{\xAxisY-\tickSz})
            node[below,font=\footnotesize]{\i};

    % x-axis label
    \node[below,font=\footnotesize]
        at (3*\xSc,{\xAxisY-0.8})
        {Cost};

    % y-axis ticks: 79 to 91
    \foreach \i in {79,...,91}
        \draw
            (\xMin*\xSc,{(\i-\yBase)/\yCompress*\ySc})
            --
            ({\xMin*\xSc-\tickSz},
             {(\i-\yBase)/\yCompress*\ySc})
            node[left,font=\footnotesize]{\i};

    % y-axis
    \draw[->]
        (\xMin*\xSc,\xAxisY)
        --
        (\xMin*\xSc,{(91.5-\yBase)/\yCompress*\ySc});

    % y-axis label
    \node[rotate=90,font=\footnotesize]
        at ({\xMin*\xSc-1.6},
            {(85-\yBase)/\yCompress*\ySc})
        {Total Payoff};

    %--------------------------------------------------
    % Unknown
    %--------------------------------------------------

    % Payoff = 92.5 - 3c
    \draw[ForestGreen, line width=1.2pt]
        (\xMin*\xSc,
         {((92.5-3*\xMin)-\yBase)/\yCompress*\ySc})
        --
        (2.92*\xSc,
         {((92.5-3*2.92)-\yBase)/\yCompress*\ySc});

    % Payoff = 80
    \draw[ForestGreen, line width=1.2pt]
        (2.92*\xSc,\bWlf)
        --
        (\xMax*\xSc,\bWlf);

    %--------------------------------------------------
    % Known
    %--------------------------------------------------

    % Payoff = 92.5 - 3c
    \draw[Maroon, line width=1.2pt]
        (\xMin*\xSc,
         {((92.5-3*\xMin)-\yBase)/\yCompress*\ySc+0.1})
        --
        (1.25*\xSc,
         {((92.5-3*1.25)-\yBase)/\yCompress*\ySc+0.1});

    % Payoff = 91.25 - 2c
    \draw[Maroon, line width=1.2pt]
        (1.25*\xSc,
         {((91.25-2*1.25)-\yBase)/\yCompress*\ySc+0.1})
        --
        (2.5*\xSc,
         {((91.25-2*2.5)-\yBase)/\yCompress*\ySc+0.1});

    % Payoff = 87.5 - c
    \draw[Maroon, line width=1.2pt]
        (2.5*\xSc,
         {((87.5-2.5)-\yBase)/\yCompress*\ySc+0.1})
        --
        (5*\xSc,
         {((87.5-5)-\yBase)/\yCompress*\ySc+0.1});

    % Payoff = 80, shifted slightly for visibility
    \draw[Maroon, line width=1.2pt]
        (5*\xSc,{\bWlf+0.1})
        --
        (\xMax*\xSc,{\bWlf+0.1});

    %--------------------------------------------------
    % Coarse
    %--------------------------------------------------

    % Payoff = 92.5 - 3c
    \draw[NavyBlue, line width=1.2pt]
        (\xMin*\xSc,
         {((92.5-3*\xMin)-\yBase)/\yCompress*\ySc+0.2})
        --
        (1.25*\xSc,
         {((92.5-3*1.25)-\yBase)/\yCompress*\ySc+0.2});

    % Payoff = 91.25 - 2c
    \draw[NavyBlue, line width=1.2pt]
        (1.25*\xSc,
         {((91.25-2*1.25)-\yBase)/\yCompress*\ySc+0.2})
        --
        (3.75*\xSc,
         {((91.25-2*3.75)-\yBase)/\yCompress*\ySc+0.2});

    % Payoff = 80, shifted slightly for visibility
    \draw[NavyBlue, line width=1.2pt]
        (3.75*\xSc,{\bWlf+0.2})
        --
        (\xMax*\xSc,{\bWlf+0.2});

    %--------------------------------------------------
    % Dotted cutoff lines
    %--------------------------------------------------

    \foreach \i in {1.25,2.5,2.92,3.75,5}
        \draw[dotted,->]
            (\i*\xSc,{(91-\yBase)/\yCompress*\ySc})
            node[above,font=\footnotesize]{\i}
            --
            (\i*\xSc,\xAxisY);

    %--------------------------------------------------
    % Legend
    %--------------------------------------------------

    \draw[ForestGreen, line width=1.2pt]
        (1.0*\xSc,{\xAxisY-2.5})
        --
        (1.5*\xSc,{\xAxisY-2.5})
        node[right,font=\footnotesize]{Unknown};

    \draw[Maroon, line width=1.2pt]
        (2.6*\xSc,{\xAxisY-2.5})
        --
        (3.1*\xSc,{\xAxisY-2.5})
        node[right,font=\footnotesize]{Known};

    \draw[NavyBlue, line width=1.2pt]
        (4.0*\xSc,{\xAxisY-2.5})
        --
        (4.5*\xSc,{\xAxisY-2.5})
        node[right,font=\footnotesize]{Coarse};

\end{tikzpicture}

\caption{Theoretical Prediction: Total Student Payoff under Symmetric Costs}
\label{fig:Welfare_Exp 1_th_sym}

\begin{tabnotes}
The figure shows how equilibrium total payoffs compare across the three treatments when learning costs are symmetric. Whenever multiple equilibria exist, the maximum-information, and therefore highest-payoff, equilibrium is selected. The dotted vertical lines indicate points at which some students change their learning decisions in one of the treatments. In equilibrium, all students acquire information in all treatments when the cost is below 1.25, and no students acquire information in any treatment when the cost is above 5. To ensure all treatments are visible, the lines have been slightly offset vertically; in regions where the lines run parallel, the theoretical payoffs are mathematically identical.
\end{tabnotes}
\end{figure}

This trade-off is apparent in our experimental setting. Under \textit{Coarse}, two of the three students (Students 1/2) know only that they are Student 1 or Student 2 with equal probability. They therefore evaluate their learning benefits by averaging across these two contingencies. If the other Student 1/2 acquires information, the learning benefit of a Student 1/2 is $(5+2.5)/2=3.75$, which exceeds Student 2's learning benefit of 2.5 under \textit{Known}. Hence, for some information costs, both higher-priority students acquire information under \textit{Coarse}, whereas only Student 1 does so under \textit{Known}, resulting in higher social welfare under \textit{Coarse}. Intuitively, precise knowledge about priorities under \textit{Known} gives Student 1 an ``unnecessarily'' high learning benefit, but Student 2's learning benefit falls sharply. In contrast, the less precise knowledge under \textit{Coarse} prompts each Student 1/2 to pool her learning benefits across the two possible priority positions. This sustains strong incentives for both students and induces more information acquisition.

Following this intuition, we show that greater priority transparency does not always improve welfare. Figure \ref{fig:Welfare_Exp 1_th_sym} illustrates how welfare comparisons across the three treatments depend on learning costs when costs are symmetric.\footnote{To make meaningful comparisons across the three treatments, we focus on ex-ante welfare throughout the paper. In all welfare figures, the y-axis origin is set at 79, since random assignment with no students acquiring information yields a benchmark total student payoff of 80.} When all students have the same learning cost, \textit{Coarse} outperforms \textit{Known} when costs lie between 2.5 and 3.75, whereas \textit{Known} yields weakly higher welfare otherwise.\footnote{Figure \ref{fig:Welfare_theory_sym_min} in \ref{sec:supp_figures} presents the theoretical prediction under the minimum-information equilibrium for all symmetric costs in $[1,5]$ on a 0.2 grid. Under the minimum-information equilibrium, \textit{Known} weakly outperforms \textit{Coarse}, which in turn weakly outperforms \textit{Unknown}, for all symmetric costs.} Both \textit{Known} and \textit{Coarse} yield weakly higher welfare than \textit{Unknown} for all symmetric costs. With asymmetric costs, \textit{Unknown} can yield higher welfare than \textit{Known} for some cost combinations; see Figure \ref{fig:Welfare_theory_all} in \ref{sec:supp_figures}.

In summary, our theoretical analysis shows that greater priority transparency can improve information acquisition and welfare in many cases. However, the relationship is not monotonic: coarser priority disclosure can sometimes yield higher welfare.

\section{Experimental Design}
As indicated by the theoretical analysis, students must navigate several steps when deciding whether to acquire information. Given information costs, they first form beliefs about other students' learning decisions. They must then assess how these decisions affect their admission chances before ultimately making a learning decision of their own.

We design three experiments to unpack this process. Experiment 2 examines the modeled game, in which subjects make learning decisions based on their own cost and the distribution of other students' costs. Experiment 1 uses robots to isolate the link between other students' learning decisions and a subject's own learning decision, thereby examining subjects' understanding of strategic complementarity in learning. Moreover, by fixing beliefs about other students' learning decisions, Experiment 1 allows us to control for the specific equilibrium being played in the data analysis. Experiment 3 removes strategic interaction by directly providing subjects with their admission chances. This allows us to isolate how subjects translate admission-chance information into learning decisions, while accounting for non-strategic factors that affect decisions over lotteries, such as risk preferences and other features of subjects’ utility functions. Together, the three experiments identify the key challenges involved in learning decisions and therefore have important policy implications.

In each experiment, we use a multiple price list to elicit subjects' willingness to pay (WTP) for information about their own preferences. We order the experiments to help subjects gradually understand the strategic considerations in the game. Experiment 1 comes first because it fixes other students' learning decisions, allowing subjects to focus on their own best responses. Experiment 2 follows because subjects must instead infer other students' learning decisions from their learning costs. Experiment 3 is conducted last because it removes strategic considerations altogether and serves as a simple elicitation task.

\subsection{Experiment 1: WTP given others' learning decisions}

In Experiment 1, each game is played by one human subject and two ``robots,'' representing the other two students in the market. Subjects are directly informed of the robots' learning decisions and are told that these decisions indicate how a robot will pick between Schools A and B in the admission procedure. Specifically, if a robot has chosen ``Learn,'' it will choose B if B is worth 40 and choose A if B is worth 10. Each case occurs with probability 0.5. If the robot has chosen ``Not Learn,'' it will always choose A.

The experiment begins with 15 practice rounds. In every practice round, each subject is first informed of her own learning cost and the learning decisions of the two robots; costs and decisions are randomly drawn each round. Subjects are asked to make learning decisions based on this information. 

After the information-acquisition stage, the admission procedure starts and subjects choose schools in the order of their priorities. At the end of each round, subjects are informed of the admission result and their hypothetical payoff for that round.

The experimental rounds involve the same decisions as the practice rounds. However, instead of having subjects respond to a single learning cost in each round, we use a multiple price list to elicit their WTP for information. After informing subjects of the robots' learning decisions, we ask them to plan their learning strategy in advance for different learning costs. Table \ref{tab:WTP Elicitation} is displayed on the screen together with the question: ``At which learning cost would you switch from `Learn' to `Not Learn'?''\footnote{We are primarily interested in costs ranging from 1.25 to 5 to make non-trivial comparisons between treatments. Table \ref{tab:WTP Elicitation} covers a wider range aimed at mitigating potential boundary effects in WTP elicitation.} After a subject enters a switch point and clicks a button labeled ``Fill Table,'' the computer chooses ``Learn'' for costs smaller than the switch point and chooses ``Not Learn'' for costs equal to or larger than the switch point. Subjects can also enter 100 to choose ``Learn'' for every cost in the table or enter 0 to choose ``Not Learn'' for every cost. Next, for each subject, one cost is randomly selected as the actual learning cost for the round. If the selected cost is smaller than the switch point, the subject learns the School-B payoff and pays the selected cost. Otherwise, the subject proceeds directly to the admission procedure. 

\begin{table}[htbp]
\centering
\caption{WTP Elicitation}
\begin{tabular}{c|c}
\hline
\textit{If your learning cost is} & \textit{do you want to learn about School B?} \\
\hline
0.2 & \(\square\) Learn \hspace{0.4cm}    \(\square\) Not Learn \\
\hline
0.4 & \(\square\) Learn \hspace{0.4cm}    \(\square\) Not Learn \\
\hline
0.6 & \(\square\) Learn \hspace{0.4cm}    \(\square\) Not Learn \\
\hline
\(\vdots\) & \(\vdots\)  \\
\hline
5.8 & \(\square\) Learn \hspace{0.4cm}    \(\square\) Not Learn \\
\hline
6 & \(\square\) Learn \hspace{0.4cm}    \(\square\) Not Learn \\
\hline
\end{tabular}
\label{tab:WTP Elicitation}
\begin{tabnotes}
    The table presents the multiple price list used to elicit subjects' WTP for information. Rather than choosing between ``Learn'' and  ``Not Learn'' for every cost in the table, subjects select a switch point, and the computer fills the table with ``Learn'' for costs below the switch point and ``Not Learn'' for costs equal to or above the switch point.
\end{tabnotes}
\end{table}

Across rounds, subjects face all possible scenarios of the other two students' learning decisions. In \textit{Known}, each subject can distinguish between the two other students and therefore faces four rounds: (Learn, Learn), (Learn, Not Learn), (Not Learn, Learn), and (Not Learn, Not Learn). The same applies to Students 1/2 in \textit{Coarse}, who can distinguish between the other Student 1/2 and Student 3. By contrast, Student 3 in \textit{Coarse}, and all students in \textit{Unknown}, cannot distinguish between the two other students and therefore face three rounds, corresponding to whether zero, one, or two of the other students learn.

\subsection{Experiment 2: WTP given others' learning costs}
In Experiment 2, subjects are divided into groups of three. In every round, each student is informed of the learning costs of the other two students and is asked to plan her learning strategy for different possible learning costs of her own.\footnote{Subjects are told only that the two costs are assigned to the other two students with equal probability; they do not know which cost is assigned to which student.} We also elicit beliefs by asking each student to predict the learning decisions of the other two students based on their costs in that round. 

There are six rounds. Across these rounds, each subject is informed that the other two students have the following costs: \{1.8, 2.6\}, \{1.8, 3\}, \{2.6, 3\}, \{1.8, 1.8\}, \{2.6, 2.6\}, \{3, 3\}. As detailed in Section \ref{sec:exp_results} and \ref{sec:Recombinant}, recombinant estimation allows us to examine all 27 combinations of the three costs \{1.8, 2.6, 3\}. The costs are chosen so that the theory predicts a non-monotonic welfare effect of priority transparency. For all but two cost combinations, the predicted ranking is $\textit{Coarse} \geq \textit{Known} \geq \textit{Unknown}$. The two exceptions are $(1.8, 2.6, 1.8)$ and $(2.6, 2.6, 1.8)$, for which the theory predicts $\textit{Coarse} \geq \textit{Unknown} \geq \textit{Known}$. This design allows us to test whether the predicted non-monotonic effect arises in the experiment.

In each round, we elicit WTP and beliefs only; the admission procedure is not implemented, and subjects receive no feedback. After the sixth round, the computer randomly selects one round for payment. Using the WTP table submitted for that round, the computer determines whether each subject acquired information. Subjects who acquired information are shown their realized School-B payoff. The admission procedure is then implemented, and subjects make their school decisions for the selected round. A subject's payoff equals her realized school payoff, minus the learning cost if she acquired information, plus 1 AUD if her predictions for both other students in the selected round are correct.

\subsection{Experiment 3: WTP given admission chances}\label{sec:exp3}

In Experiment 3, subjects no longer play against robots or other human subjects. Instead, each subject is directly informed of the probability that School A will be available to her and is asked to plan her learning strategy for different learning costs based on this information.

Across rounds, subjects are informed of admission probabilities corresponding to scenarios with different learning decisions of other students in Experiment 1, as shown in Column 4 of Table \ref{tab:learning incentives}. For instance, subjects in \textit{Unknown} are presented with admission probabilities of 58.3\%, 41.7\%, and 33.3\% in Experiment 3, corresponding to cases in which two, one, and none of the other students acquire information in Experiment 1, respectively. In addition, all subjects in all treatments face a 0\% admission probability in one round. This allows us to measure ``curiosity,'' defined as a subject's WTP for non-instrumental information.\footnote{\cite{chen2021information} measure curiosity using subjects' WTP to learn the realization of a lottery in the Holt--Laury risk-preference elicitation task. They identify curiosity as an important predictor of information-acquisition behavior.} Each subject faces four rounds with different admission probabilities.\footnote{Some students have fewer relevant admission probabilities. For example, Student 1 in \textit{Known} has an admission probability of 1 regardless of the other students' learning decisions. In such cases, we included additional admission probabilities that were not relevant to that student. Responses to these additional questions are not used in the analysis.}

In each round, we elicit WTP without providing feedback. After the fourth round, the computer randomly selects one round and draws a learning cost. The subject's stated learning decision at that cost determines whether she receives information about her preferences. The computer then randomly determines whether School A is available according to the probability specified in the selected round. If School A is available, the subject chooses between Schools A and B; otherwise, she is automatically assigned to School B.

\subsection{Experimental Procedure}

The experiment was conducted in the Experimental Economics Laboratory of the University of Melbourne $\text{(E}^{2}\text{MU)}$ and programmed in z-Tree. Upon entering the lab, participants received the instructions for Experiment 1. Additional instructions were provided before the start of Experiments 2 and 3. We employed a between-subject design, with each subject participating in only one of the three treatments. Depending on the treatment, subjects received no, partial, or full knowledge about their priorities at the beginning of Experiment 1, and this knowledge remained unchanged throughout all three experiments.\footnote{In \textit{Unknown}, student priorities were randomly drawn at the beginning of every round in each experiment. In \textit{Known}, each subject retained the same priority throughout all three experiments. In \textit{Coarse}, the top two students remained in the top-two group throughout, but their exact priorities (1 or 2) were randomly determined at the beginning of every round. Student 3 retained the same priority throughout.} At the end of each session, one round from one experiment was randomly selected for payment. In total, we conducted nine sessions with 180 participants. Treatments \textit{Unknown}, \textit{Known}, and \textit{Coarse} each had three sessions, with 63, 63, and 54 subjects, respectively. Average earnings were approximately 32.6 AUD, including a participation fee of 10 AUD. Sessions lasted approximately 90 minutes.

\subsection{Hypotheses}

Our experimental design allows us to test the following hypotheses derived from the theoretical framework.

\begin{hypothesis}
(\textit{Welfare}) Total student payoff follows model predictions: (i) the welfare ranking of \textit{Known} relative to \textit{Unknown} and \textit{Coarse} depends on the distribution of information costs and equilibrium selection; and (ii) \textit{Coarse} weakly dominates \textit{Unknown} for all cost combinations.
\end{hypothesis}

\begin{hypothesis}
(\textit{Learning decisions}) Students' WTP for information about their own preferences follows the model predictions specified in Table \ref{tab:learning incentives}.
\end{hypothesis}

\begin{hypothesis}
(\textit{Strategic complementarity}) A student's WTP for information increases as more students relevant to her choose to acquire information.
\end{hypothesis}

\begin{hypothesis}
(\textit{Beliefs}) Students form correct beliefs about other students' learning decisions based on the distribution of information costs.
\end{hypothesis}

\begin{hypothesis}
(\textit{School decisions}) Students' choices between Schools A and B follow model predictions: they choose School B after observing a School-B payoff of 40, and choose School A after observing a payoff of 10 or when they do not observe their School-B payoff.
\end{hypothesis}

\section{Experimental results}
\label{sec:exp_results}

In this section, we examine whether the experimental evidence supports the hypotheses above. We begin with student welfare, the primary market-level outcome, and then examine individual decision-making through WTP, strategic complementarity, beliefs, and school decisions to explain the observed treatment differences and behavioral biases. Unless otherwise noted, statistical significance refers to the 5\% level. When reporting experimental results, ($>$) and ($<$) indicate statistically significant differences, while ($=$) indicates differences that are not statistically significant.

\subsection{Welfare}

We use total student payoff in each experimental market as our measure of social welfare. All three experiments are designed to elicit subjects' strategies in different scenarios, without interaction among subjects during the elicitation stage. In Experiment 1, subjects interact only with robots, not with each other, while Experiment 3 involves only individual decision-making. In Experiment 2, subjects receive no feedback until the end of the elicitation stage; although they subsequently participate in sequential RSD to determine their payoff for the selected round, we do not use these realized outcomes in the analysis. This design allows us to use the recombinant estimator to compare welfare across the three treatments (\citealp{mullin2006recombinant,abrevaya2008recombinant}). In recombinant estimation, we combine subjects into different groups of three and use their elicited strategies to compute the outcomes that would have occurred under such groupings. \ref{sec:Recombinant} provides a detailed description of the estimation approach.

We begin the welfare analysis with Experiment 2 because it corresponds to the modeled game and most closely resembles the game that students face in school choice. In this experiment, each subject faces two other subjects and makes decisions based on the distribution of the other students' learning costs.
Figure \ref{fig:Welfare_Exp 2} presents total student payoff averaged across markets under each treatment. According to the theory, for all but two of the cost combinations examined in Experiment 2, the welfare ranking should be $\textit{Coarse} \geq \textit{Known} \geq \textit{Unknown}$. The two exceptions are $(1.8, 2.6, 1.8)$ and $(2.6, 2.6, 1.8)$, for which the theory predicts $\textit{Coarse} \geq \textit{Unknown} \geq \textit{Known}$. The experimental evidence in Figure~\ref{fig:Welfare_Exp 2_exp} clearly contradicts this prediction and instead favors \textit{Known}. Table \ref{tab:welfare_exp2_costs} in \ref{sec:supp_figures} reports total student payoffs and pairwise treatment comparisons in Experiment 2 for each cost combination. Across all cost combinations, the experimental welfare ranking is
\[
\textit{Known} \geq \textit{Coarse} \geq \textit{Unknown}.
\]
Moreover, compared with the theoretical payoff levels shown in Figure \ref{fig:Welfare_Exp 2_theo}, observed payoffs are significantly lower in all treatments. This suggests that subjects may exhibit systematic behavioral deviations or make systematic strategic mistakes, which we examine in the subsequent sections.

\begin{figure}[htbp]
  \centering
  \begin{subfigure}{0.49\textwidth}
      \includegraphics[width=\linewidth]{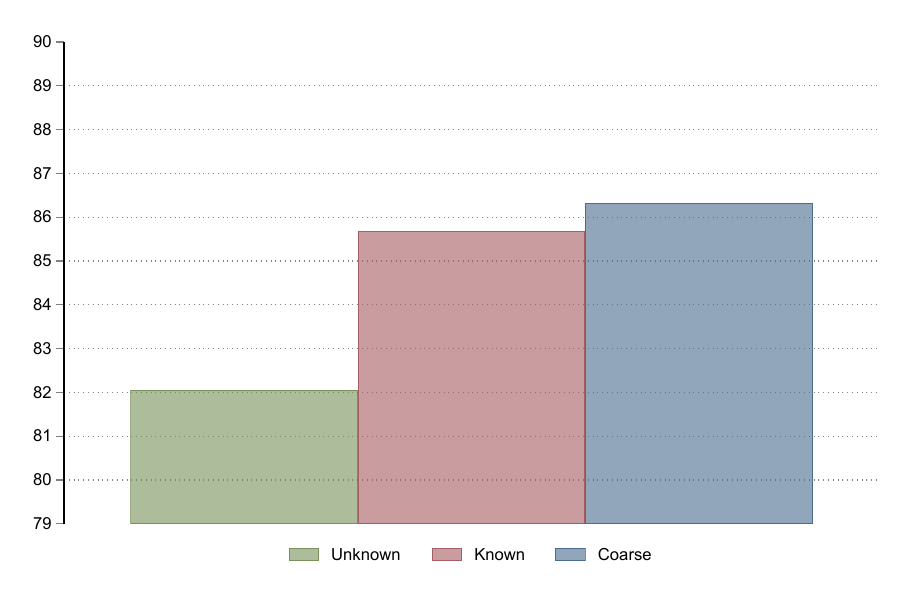}
    \caption{Theory}
    \label{fig:Welfare_Exp 2_theo}
  \end{subfigure}
  \hfill
  \begin{subfigure}{0.49\textwidth}
    \includegraphics[width=\linewidth]{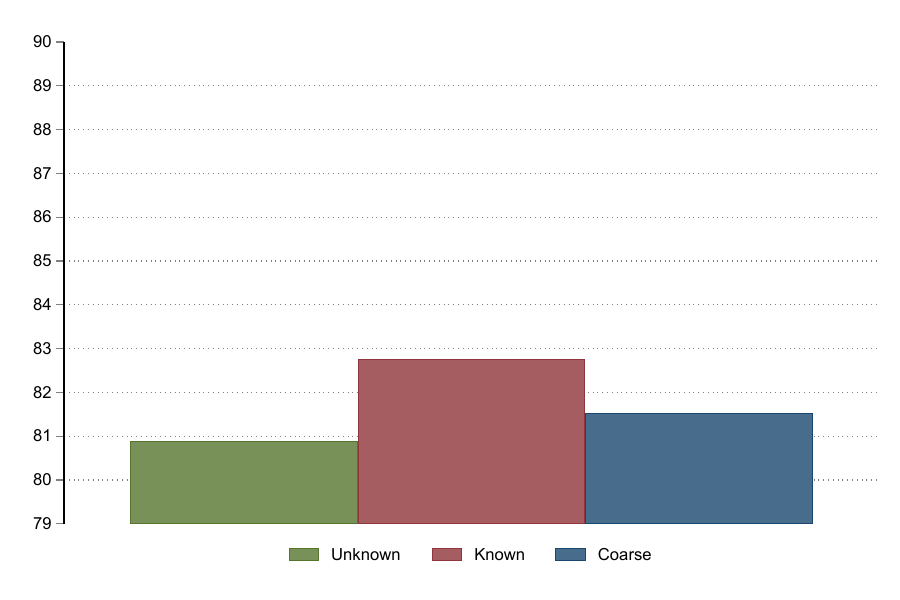}
   \caption{Experiment 2}
    \label{fig:Welfare_Exp 2_exp}
  \end{subfigure}
  \caption{Total Student Payoff in Experiment 2}
  \label{fig:Welfare_Exp 2}
  \begin{tabnotes}
      The figures show total student payoff averaged across markets under each treatment. Panel (a) presents the theoretical prediction under maximum-information equilibrium selection, and Panel (b) presents the experimental results.
  \end{tabnotes}
\end{figure}

Because \textit{Coarse} and \textit{Unknown} exhibit multiple equilibria, lower welfare in these treatments could be the result of equilibrium selection. Experiment 1 allows us to rule out this explanation by informing subjects of the other students' learning decisions, which are implemented by robots, thereby fixing their beliefs and controlling for equilibrium selection. Specifically, we map each Experiment-2 cost combination into the corresponding robot learning decisions in Experiment 1 under the maximum-information equilibrium. We then calculate welfare outcomes using recombinant estimation, assuming correct play in the sequential RSD. Table~\ref{tab:welfare_exp1_exp2cost_max} in \ref{sec:supp_figures} reports total student payoffs and pairwise treatment comparisons. Across all cost combinations, the experimental welfare ranking \textit{Known} $\geq$ \textit{Coarse} $\geq$ \textit{Unknown} remains, suggesting that equilibrium selection is unlikely to explain the discrepancy between the theoretical predictions and the experimental evidence observed in Experiment 2.

\begin{figure}[htbp]
 \centering
 \includegraphics[width=0.8\linewidth]{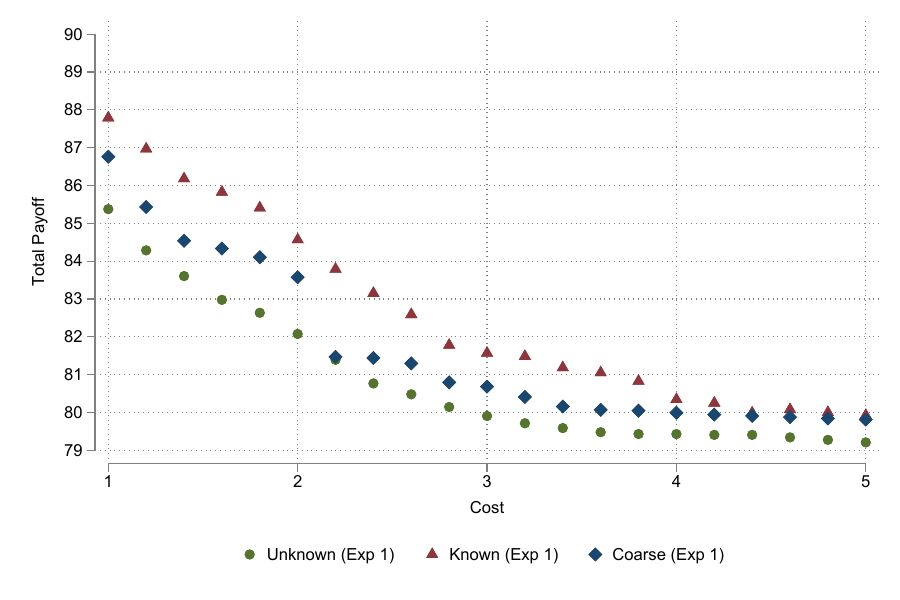}
 \caption{Total Student Payoff in Experiment 1 (Symmetric Costs)}
 \label{fig:Welfare_Exp 1_sym}
 \begin{tabnotes}
     The figure reports total student payoff in Experiment 1 for all symmetric costs in $[1,5]$ on a 0.2 grid, fixing beliefs at the maximum-information equilibrium.
 \end{tabnotes}
\end{figure}

While Experiment~2 focuses on 27 specific cost combinations, Experiment 1 allows us to test the robustness of our results to cost selection by extending the analysis to any cost triple in $[1,5]$ on a 0.2 grid. Figure \ref{fig:Welfare_Exp 1_sym} presents the welfare comparison for symmetric costs, fixing beliefs at the maximum-information equilibrium. It shows that the same welfare ranking, \textit{Known} $\geq$ \textit{Coarse} $\geq$ \textit{Unknown}, holds for all symmetric cost levels.\footnote{Experiment 1 also allows us to examine the minimum-information equilibrium. Figure~\ref{fig:Welfare_sym_min} in \ref{sec:supp_figures} presents the welfare comparison for symmetric costs, showing both the theoretical prediction under minimum-information equilibrium selection and the corresponding Experiment 1 results when beliefs are fixed at the minimum-information equilibrium. In this case, the experimental welfare ranking also largely follows the same pattern, with \textit{Known} weakly higher than \textit{Coarse}, which in turn is weakly higher than \textit{Unknown}, for almost all costs.} This conclusion continues to hold when we allow for asymmetric costs. Figure \ref{fig:Welfare_Exp 1_all} groups cost combinations according to the theoretically predicted treatment ranking.\footnote{The corresponding theoretical predictions are reported in Figure~\ref{fig:Welfare_theory_all} in \ref{sec:supp_figures}.} Across all groups, the same experimental pattern emerges. These results suggest that the observed welfare ranking cannot be attributed to equilibrium selection and is robust across cost combinations.

\begin{figure}[htbp]
 \centering
 \includegraphics[width=1\linewidth]{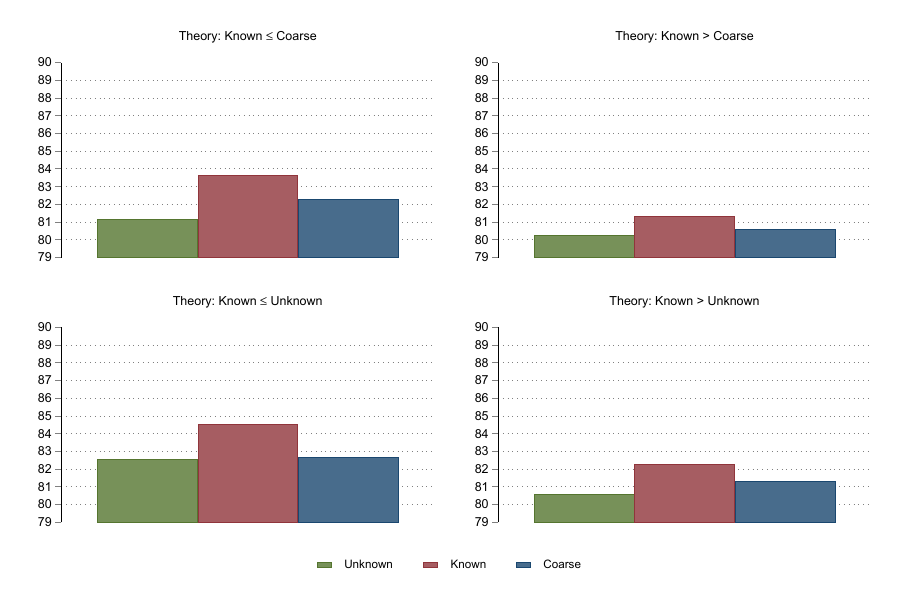}
 \caption{Total Student Payoff in Experiment 1 (All Costs)}
 \label{fig:Welfare_Exp 1_all}
 \begin{tabnotes}
     The figure reports total student payoff in Experiment 1 for all cost combinations in $[1,5]^3$ on a 0.2 grid, fixing beliefs at the maximum-information equilibrium. Cost combinations are grouped according to the theoretically predicted treatment ranking. The corresponding theoretical predictions are reported in Figure~\ref{fig:Welfare_theory_all} in \ref{sec:supp_figures}.
 \end{tabnotes}
\end{figure}

Our findings on student welfare are summarized below.
\medskip{}

\noindent \textbf{Result 1 (Welfare)}

\noindent \textit{(1) In terms of total student payoff, \textit{Known} $\geq$ \textit{Coarse} $\geq$ \textit{Unknown}.}

\noindent \textit{(2) In all treatments, student payoffs are significantly lower than the theoretical predictions.}

\medskip{}

These findings support part (ii) of Hypothesis 1: \textit{Coarse} generally outperforms \textit{Unknown}. However, the results largely contradict part (i) and reject the prediction that less knowledge about priorities can benefit students. Instead, we find that more precise and definite knowledge about priorities generally improves student welfare. We rule out equilibrium selection as an explanation and, in subsequent sections, analyze individual behavior to explore other potential explanations for the welfare results, including subjects' understanding of strategic considerations in different treatments and behavioral patterns in response to admission chances.

\subsection{WTP}

This section examines learning decisions by comparing the subjects' WTP with the theoretical predictions. We focus on Experiment 1 because it allows us to control for equilibrium selection and derive clear theoretical benchmarks.

We use deviation in WTP, defined as the difference between experimental WTP and theoretical WTP, to assess overall tendencies in information-acquisition behavior. Because positive and negative deviations offset each other, we also use absolute deviation, defined as the absolute value of this difference, to measure the magnitude of deviations from the theoretical prediction. Figure \ref{fig:WTP_Exp 1} reports the average deviation (``Unknown,'' ``Known,'' and ``Coarse'') along with the average absolute deviation (``Unknown (Abs),'' ``Known (Abs),'' and ``Coarse (Abs)'') across three treatments.

The figure indicates a tendency to under-acquire information in all three treatments: the average deviation is negative in \textit{Unknown} ($p=0.043$), \textit{Known} ($p=0.066$), and \textit{Coarse} ($p<0.01$). Although the evidence is weaker in \textit{Known}, the average deviation remains negative. We find no significant treatment differences in the magnitude of deviations from the theoretical prediction.

The literature provides mixed experimental evidence on whether subjects tend to over- or under-acquire information. \citet{hakimov2023costly} and \citet{descamps2022learning} find that subjects acquire too much information when information costs are high and too little when costs are low. In contrast, other studies document overall over-search, including \citet{chen2021information} in school choice, \citet{bhattacharya2017voting} in voting, and \citet{gretschko2015excess} in auctions. One potential explanation for the overall lower level of information acquisition in our study is that subjects had extensive training before the WTP elicitation. Figure \ref{fig:deviation_practice} in \ref{sec:supp_figures} reports the average deviation of learning decisions from the theoretical predictions during the practice rounds. Subjects initially over-search, but the level of over-search declines over time ($p=0.022$), and they exhibit under-search in some later rounds.\footnote{We code ``Learn'' as 1 and ``Not Learn'' as 0, and define deviation as the experimental decision minus the theoretical prediction. Thus, zero indicates correct learning, positive values indicate over-search on average, and negative values indicate under-search on average. The $p$-value is derived from an ordered logit regression of the deviation in learning on round number, controlling for curiosity and demographic attributes. Standard errors are clustered at the subject level.} Risk aversion provides another possible explanation, since it can reduce subjects’ incentives to acquire information in our setting. However, as discussed in Section~\ref{sec:exp3}, Experiment 3 allows us to account for risk aversion and other features of subjects’ utility functions; our results are robust to this control.

\begin{figure}[htbp]
 \centering
 \includegraphics[width=0.8\linewidth]{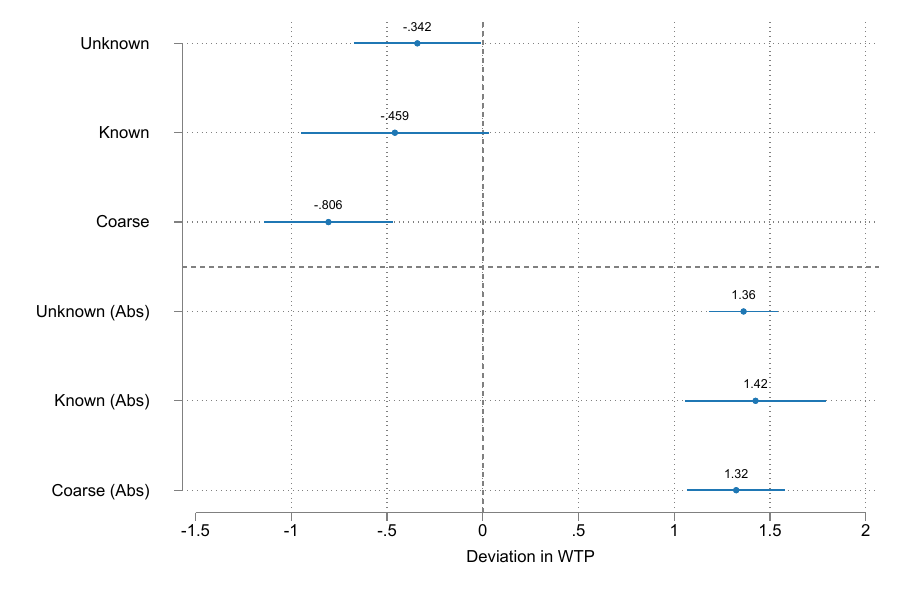}
 \caption{Deviation in WTP in Experiment 1}
 \label{fig:WTP_Exp 1}
   \begin{tabnotes}
 ``Unknown,'' ``Known,'' and ``Coarse'' refer to the deviation in WTP (the difference between experimental WTP and theoretical WTP) in the corresponding treatments. ``Unknown (Abs),'' ``Known (Abs),'' and ``Coarse (Abs)'' refer to the absolute deviation in WTP (the absolute value of the difference between experimental WTP and theoretical WTP) in corresponding treatments. Average levels are presented with 95\% confidence intervals, adjusting for censoring from below and above, and standard errors are clustered at the level of individual subjects.
  \end{tabnotes}
\end{figure}

\begin{table}[htbp]
\centering
\caption{Determinants of WTP and Absolute Deviation in WTP}
\def\sym#1{\ifmmode^{#1}\else\(^{#1}\)\fi}
\begin{tabular*}{1\hsize}{@{\hskip\tabcolsep\extracolsep\fill}l*{2}{cc}}
\hline\hline
            &\multicolumn{2}{c}{(1)}           &\multicolumn{2}{c}{(2)}           \\
            &\multicolumn{2}{c}{WTP}        &\multicolumn{2}{c}{Absolute Deviation in WTP}  \\
\hline
%main        &                     &            &                     &            \\
\textit{Known}&     0.00678         &     (0.262)&       0.121         &     (0.199)\\
\textit{Coarse}&       -0.335         &     (0.242)&     -0.0521         &     (0.161)\\
\textit{Curiosity}   &       0.435\sym{***}&     (0.114)&       0.286\sym{***}&    (0.0821)\\
\textit{Female}      &       -0.284         &     (0.231)&       0.221         &     (0.177)\\
\textit{Graduate}        &      -0.221         &     (0.278)&      -0.197         &     (0.223)\\
\textit{Age}         &      0.0269         &    (0.0229)&     0.00735         &    (0.0214)\\
Constant      &        1.153\sym{**}  &     (0.518)&       0.905\sym{*}         &     (0.479)\\
\hline
\(N\)       &         639         &            &         639         &            \\
\hline\hline
\end{tabular*} 
    \begin{tabnotes} 
        In the table, \textit{Known} and \textit{Coarse} are dummy variables that equal to one for the respective treatments and zero otherwise. \textit{Curiosity} is the WTP for information when the admission chance is zero in Experiment 3. \textit{Female} is a dummy that equals one for female subjects and zero otherwise. \textit{Graduate} is a dummy that equals one for graduate students and zero for undergraduate students. \textit{Age} represents the age of subjects, ranging from 18 to 52 years old. Tobit models are adopted with censoring from below and above. Standard errors are clustered at the level of individual subjects. \sym{*} \(p<0.1\), \sym{**} \(p<0.05\), \sym{***} \(p<0.01\).
    \end{tabnotes}
\label{tab:regression WTP & abs deviation}
\end{table}

Next, we explore factors that may influence subjects' learning decisions and their deviations from the theoretical prediction. Because WTP is elicited over a bounded interval, we use Tobit regressions to account for censoring from below and above. Standard errors are clustered at the subject level. In Table \ref{tab:regression WTP & abs deviation}, regression (1) examines subjects' WTP in Experiment 1, regressing it on treatment dummies, a measure of curiosity, and demographic attributes, including gender, age, and whether subjects are graduate or undergraduate students.

Overall, the results reveal no significant effect of treatment dummies or demographic attributes on subjects' WTP for information. The curiosity measure comes from Experiment 3, where we elicit subjects' WTP for information when their admission chance at School A is 0\%, so that information has no instrumental value. Approximately 17\% of subjects exhibit curiosity by assigning a positive WTP to this non-instrumental information. The regression results indicate a positive relationship between curiosity and WTP in Experiment 1, suggesting that subjects with higher curiosity tend to acquire more information.\footnote{Table \ref{tab:regression WTP_treatment} in \ref{sec:supp_figures} shows that the effect of curiosity is primarily driven by \textit{Known}.} This finding is consistent with \citet{chen2021information}, who use a different measure of curiosity in a different information-acquisition environment. Together, the evidence suggests that curiosity is an important predictor of information-acquisition behavior across settings and measurement approaches.

Regression (2) examines the magnitude of deviations from the theoretical prediction, measured as the absolute difference between subjects' WTP and the theoretical prediction. The results indicate that subjects with higher curiosity deviate further from the theoretical prediction. We find no significant effects of treatment dummies or demographic attributes.

\begin{figure}[htbp]
 \centering
 \includegraphics[width=0.8\linewidth]{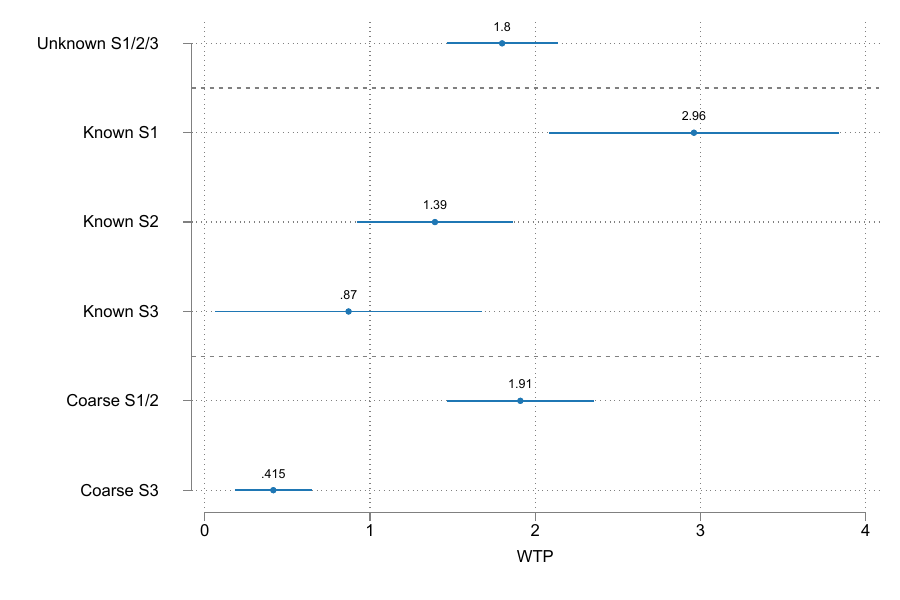}
 \caption{WTP by Treatments and Roles}
 \label{fig:WTP_by_role}
    \begin{tabnotes}
        The figure shows average WTP by treatment and priority roles. Averages are reported with 95\% confidence intervals, adjusting for censoring from below and above. Standard errors are clustered at the subject level.

  \end{tabnotes}
\end{figure}

Lastly, we examine whether subjects take their knowledge of priorities into account when making learning decisions. Figure \ref{fig:WTP_by_role} presents average WTP by treatment and priority roles. In \textit{Known}, average WTP is highest for Student 1, followed by Student 2 and then Student 3, with each difference statistically significant. In \textit{Coarse}, Student 1/2 has a significantly higher average WTP than Student 3. These patterns are confirmed by the regressions in Table \ref{tab:regression WTP_treatment} in \ref{sec:supp_figures}. Thus, although the significant deviation of average WTP from the theoretical prediction contradicts Hypothesis 2, subjects do appear to understand the basic link between priority and learning incentives: higher priority increases the value of acquiring information.

Below, we summarize our main findings concerning students' WTP for information about their own preferences.

\medskip{}

\noindent \textbf{Result 2 (WTP)}

\noindent \textit{(1) Average WTP is below the theoretical prediction in all treatments. The difference is significant at 5\% in \textit{Unknown} and \textit{Coarse}, and at 10\% in \textit{Known}.}

\noindent \textit{(2) Subjects with higher curiosity tend to have higher WTP and deviate further from the theoretical prediction.}

\noindent \textit{(3) When provided with knowledge about priorities, as in \textit{Known} and \textit{Coarse}, subjects choose higher WTP when assigned to higher-priority roles, consistent with the direction of the theoretical prediction.}

\subsection{Strategic Complementarity}

In this section, we examine whether subjects understand the externality generated by information acquisition and treat learning decisions as strategic complements. We first analyze how subjects' WTP changes in response to the learning decisions of other students in Experiment 1. We then use data from Experiment 3 to disentangle different potential explanations for the observed behavior.

To evaluate subjects' understanding of strategic complementarity in learning, we begin with a simple directional classification. We define $\Delta\text{WTP}$ as the change in a subject's WTP when relevant students---those whose learning decisions affect the subject's admission chances---switch from ``Not Learn'' to ``Learn''. We classify a subject as treating learning decisions as strategic complements if $\Delta\text{WTP}>0$, and as treating them as strategic substitutes if $\Delta\text{WTP}<0$. Subjects who do not fall into either category are classified as exhibiting inconsistent responses or no response.\footnote{There may be multiple cases for a given subject in which a relevant student switches from ``Not Learn'' to ``Learn''. We consider all such cases and classify a subject as treating learning decisions as strategic complements if she increases her WTP at least once and never decreases it. Conversely, we classify a subject as treating learning decisions as strategic substitutes if she decreases her WTP at least once and never increases it. We exclude Student 1 in \textit{Known} from this analysis because her WTP should not depend on the learning decisions of others.}\label{fn:classification}

\begin{table}[htbp]
\centering
\small
\caption{Strategic Complementarity}
\begin{tabularx}{\textwidth}{l|c|c|c|p{1.7cm}|p{1.7cm}|p{1.7cm}}
\hline
& \multicolumn{3}{c|}{\textbf{Treatment}} & \multicolumn{3}{c}{\textbf{$p$-Value for Test of Equality}} \\
\hline
& \textit{Unknown} & \textit{Known} & \textit{Coarse} & \textit{Unknown = Known} & \textit{Unknown = Coarse} & \textit{Known \newline = Coarse} \\
\hline
Strategic complements & 19\% & 40\% & 33\% & 0.016 & 0.078 & 0.471 \\
Strategic substitutes & 22\% & 12\% & 19\% & 0.179 & 0.621 & 0.376 \\
Inconsistent/no responses & 59\% & 48\% & 48\% & 0.263 & 0.252 & 0.959 \\
\hline
\end{tabularx}
\begin{tabnotes}
The table reports the proportion of subjects who treat learning decisions as strategic complements, substitutes, or have inconsistent or no responses. Student 1 in the Known treatment is excluded. The proportions in the treatment columns are calculated according to the classification rule described in footnote \ref{fn:classification}. The $p$-values in the last three columns are obtained from two-sample tests of proportions.
\end{tabnotes}
\label{tab:strategic complementarity}
\end{table}

Table \ref{tab:strategic complementarity} summarizes the distribution of subjects across these categories by treatment. The proportion of subjects who treat learning decisions as strategic complements is highest in \textit{Known}, followed by \textit{Coarse}, and lowest in \textit{Unknown}. The difference between \textit{Unknown} and \textit{Known} is significant at the 5\% level, while the difference between \textit{Unknown} and \textit{Coarse} is significant at the 10\% level. Conversely, the proportion of subjects who treat learning decisions as strategic substitutes follows the opposite pattern, although the treatment differences are not statistically significant. Notably, in \textit{Unknown}, more subjects treat learning decisions as strategic substitutes than as strategic complements, suggesting a limited understanding of learning incentives in this treatment. Taken together, these patterns suggest that subjects' understanding of strategic complementarity is strongest in \textit{Known}, followed by \textit{Coarse}, and weakest in \textit{Unknown}; this ordering aligns with the welfare results and suggests that strategic mistakes in learning decisions contribute to the observed welfare differences across treatments.

\medskip{}

\noindent \textbf{Result 3 (Strategic complementarity)}

\noindent \textit{The proportion of subjects who treat learning decisions as strategic complements is lower in  \textit{Unknown} than in \textit{Known} or \textit{Coarse}; the \textit{Unknown}--\textit{Known} difference is significant at 5\%, and the \textit{Unknown}--\textit{Coarse} difference at 10\%.}

\medskip{}

\begin{figure}[htbp]
 \centering
 \includegraphics[width=0.8\linewidth]{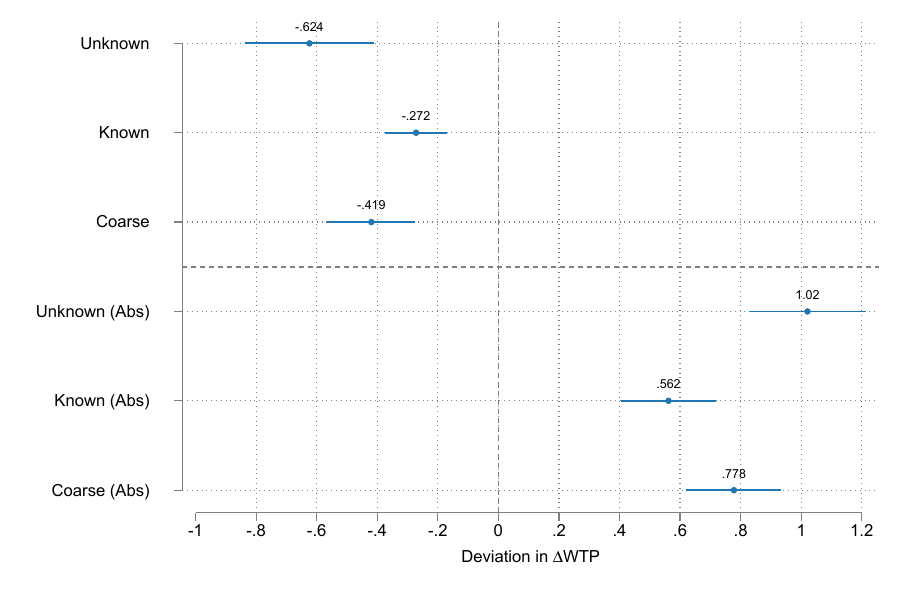}
 \caption{Deviation and Absolute Deviation in $\Delta WTP$ in Experiment 1}
 \label{fig:DWTP_Exp1_abs}
  \begin{tabnotes}
  The figure shows the difference between the experimental response and the theoretical prediction of the WTP change in response to a relevant student's switch from ``Not Learn'' to ``Learn'': $\mathcal{D}\Delta\text{WTP}_{1} = \Delta\text{WTP}_{1} - \Delta\text{WTP}_{1}^{Theory}$. ``Unknown,'' ``Known,'' and ``Coarse'' represent $\mathcal{D}\Delta\text{WTP}_{1}$, while ``Unknown (Abs),'' ``Known (Abs),'' and ``Coarse (Abs)'' represent $|\mathcal{D}\Delta\text{WTP}_{1}|$ in the corresponding treatments. Average levels are presented with 95\% confidence intervals, and standard errors are clustered at the subject level.
  \end{tabnotes}
\end{figure}

Result 3 provides support for Hypothesis 3, suggesting that subjects exhibit some understanding of strategic complementarity in learning. However, the extent of this understanding may differ across treatments. We next examine the magnitude of the subjects' responses and their deviations from the theoretical prediction. To do so, we define $\mathcal{D}\Delta\text{WTP}_{1}$, where the subscript 1 refers to Experiment 1, as the difference between a subject's actual response in Experiment 1 to a relevant student's switch from ``Not Learn'' to ``Learn''  and the theoretical prediction, so that $\mathcal{D}\Delta\text{WTP}_{1}=\Delta\text{WTP}_{1} - \Delta\text{WTP}_{1}^{Theory}$. We also consider the absolute value of this deviation because positive and negative deviations may offset each other in the signed measure.

Figure \ref{fig:DWTP_Exp1_abs} reports the average levels of $\mathcal{D}\Delta\text{WTP}_{1}$ and $|\mathcal{D}\Delta\text{WTP}_{1}|$ across the three treatments. The average value of $\mathcal{D}\Delta\text{WTP}_{1}$ is significantly below 0 in all treatments ($p<0.001$), indicating that subjects respond insufficiently to changes in others' learning decisions in Experiment 1. The magnitude of $|\mathcal{D}\Delta\text{WTP}_{1}|$ is the smallest in \textit{Known}, followed by \textit{Coarse}, and largest in \textit{Unknown} ($p<0.001$ for \textit{Unknown}--\textit{Known}; $p=0.051$ for both \textit{Known}--\textit{Coarse} and \textit{Unknown}--\textit{Coarse}).\footnote{The $p$-values are derived from linear regressions of $|\mathcal{D}\Delta\text{WTP}_{1}|$ on treatment dummies, with standard errors clustered at the subject level.} This pattern suggests that subjects make fewer strategic mistakes when more knowledge about priorities is available, which is consistent with Result 3.

We next disentangle potential sources of the deviations observed in Experiment 1. When making a learning decision given the learning decisions of others, a student must first assess how others' choices affect her admission probability and then determine her WTP given that probability. Deviations in Experiment 1 may therefore arise from mistakes in the first step: subjects may fail to understand the externality generated by information acquisition and the resulting strategic complementarity in learning. Alternatively, deviations may arise from the second step: subjects may deviate from the theoretically optimal learning decision even when admission chances are provided directly and need not be inferred. In Experiment 3, we elicit subjects' WTP for the same admission probabilities as those implied by robots' actions in Experiment 1. We therefore interpret the deviations in Experiment 3 as capturing the second-step component of the deviations observed in Experiment 1.

To assess the extent to which deviations in Experiment 1 can be explained by deviations observed in Experiment 3, we define the deviation in $\Delta\text{WTP}$ in Experiment 3 as $\mathcal{D}\Delta\text{WTP}_{3}=(\Delta\text{WTP}_{3} - \Delta\text{WTP}_{3}^{Theory})$, parallel to $\mathcal{D}\Delta\text{WTP}_{1}$. Table \ref{tab:deviation in DWTP_exp1vs3} presents the regressions of (1) $\mathcal{D}\Delta\text{WTP}_{1}$ on $\mathcal{D}\Delta\text{WTP}_{3}$, and (2) $|\mathcal{D}\Delta\text{WTP}_{1}|$ on $|\mathcal{D}\Delta\text{WTP}_{3}|$. The two regressions yield broadly consistent conclusions. For ease of interpretation, we focus on regression (2) in Table \ref{tab:deviation in DWTP_exp1vs3} in the subsequent analysis.

In the regression, $|\mathcal{D}\Delta\text{WTP}_{3}|$ has only a small and statistically insignificant marginal effect. Since Experiment 3 captures subjects' direct responses to admission chances, this insignificant coefficient suggests that, in \textit{Unknown}, deviations in Experiment 1 mainly arise from mistakes in mapping others' learning decisions into one's own admission probabilities. By contrast, the interaction terms $|\mathcal{D}\Delta\text{WTP}_{3}|\times\textit{Known}$ and $|\mathcal{D}\Delta\text{WTP}_{3}| \times\textit{Coarse}$ are large and statistically significant, indicating that deviations in WTP given admission chances account for a substantial share of the overall deviations in these treatments. That is, in \textit{Known} and \textit{Coarse}, subjects who respond suboptimally to directly provided admission chances in Experiment 3 are also more likely to respond suboptimally to changes in others’ learning decisions in Experiment 1. This effect is more pronounced in \textit{Known} than in \textit{Coarse} ($p=0.011$).

\begin{table}[htbp]
\centering
\caption{Deviation in $\Delta\text{WTP}$ in Experiments 1 and 3}
\def\sym#1{\ifmmode^{#1}\else\(^{#1}\)\fi}
\small
\begin{tabular*}{1\hsize}{@{\hskip\tabcolsep\extracolsep\fill}l*{3}{cc}}
\hline\hline
            &\multicolumn{2}{c}{(1)}           &\multicolumn{2}{c}{(2)}           &\multicolumn{2}{c}{(3)}           \\
            &\multicolumn{2}{c}{$\mathcal{D}\Delta\text{WTP}_{1}$}&\multicolumn{2}{c}{$|\mathcal{D}\Delta\text{WTP}_{1}|$}&\multicolumn{2}{c}{$|\mathcal{D}\Delta\text{WTP}_{1}-\mathcal{D}\Delta\text{WTP}_{3}|$}\\
\hline
\textit{Known}&       0.468\sym{***}&     (0.122)&      -0.694\sym{***}&     (0.146)&      -0.538\sym{***}&     (0.145)\\
\textit{Coarse}&       0.411\sym{**}  &     (0.159)&      -0.461\sym{***} &     (0.163)&      -0.347\sym{**}  &     (0.161)\\
$\mathcal{D}\Delta\text{WTP}_{3}$&       0.137         &     (0.105)&                     &            &                     &            \\
$\mathcal{D}\Delta\text{WTP}_{3}\times$\textit{Known}&       0.705\sym{***}&     (0.182)&                     &            &                     &            \\
$\mathcal{D}\Delta\text{WTP}_{3}\times$\textit{Coarse}&       0.441\sym{**}  &     (0.200)&                     &            &                     &            \\
$|\mathcal{D}\Delta\text{WTP}_{3}|$&                     &            &     -0.0265         &    (0.0985)&                     &            \\
$|\mathcal{D}\Delta\text{WTP}_{3}|\times$\textit{Known}&                     &            &       0.890\sym{***}&     (0.135)&                     &            \\
$|\mathcal{D}\Delta\text{WTP}_{3}|\times$\textit{Coarse}&                     &            &       0.466\sym{***} &     (0.168)&                     &            \\
\textit{Female}      &      0.0248         &    (0.0929)&    -0.00597         &     (0.108)&     -0.0189         &     (0.125)\\
\textit{Graduate}        &      0.0254         &    (0.0950)&      0.0794         &     (0.118)&       0.215         &     (0.145)\\
\textit{Age}         &    -0.00804         &   (0.00845)&    0.000547         &   (0.00881)&    -0.00453         &    (0.0106)\\
Constant      &      -0.410\sym{*}         &     (0.212)&       0.998\sym{***}&     (0.232)&       0.985\sym{***}&     (0.262)\\
\hline
\(N\)       &         459         &            &         459         &            &         459         &            \\
\hline\hline
\end{tabular*} 
    \begin{tabnotes} 
        The table examines how much of the deviation in Experiment 1 can be explained by the analogous deviation in Experiment 3. The variables $\mathcal{D}\Delta\text{WTP}_{1}$ and $|\mathcal{D}\Delta\text{WTP}_{1}|$ denote the deviation and absolute deviation, respectively, in $\Delta$WTP from the theoretical prediction in Experiment 1. Similarly, $\mathcal{D}\Delta\text{WTP}_{3}$ and $|\mathcal{D}\Delta\text{WTP}_{3}|$ denote the corresponding variables in Experiment 3. \textit{Known} and \textit{Coarse} are treatment dummies, with \textit{Unknown} as the omitted category. Interaction terms are defined by interacting the Experiment 3 deviation variables with the treatment dummies. \textit{Female} is a dummy equal to one for female subjects and zero otherwise. \textit{Graduate} is a dummy equal to one for graduate students and zero for undergraduate students. \textit{Age} denotes subjects' age, ranging from 18 to 52. Linear models are used for simplicity; accounting for censoring from below and above does not alter the qualitative conclusions. Standard errors are clustered at the subject level. \sym{*} \(p<0.1\), \sym{**} \(p<0.05\), \sym{***} \(p<0.01\).

    \end{tabnotes}
\label{tab:deviation in DWTP_exp1vs3}
\end{table}

To further investigate the contribution of deviations observed in Experiment 3, we subtract the Experiment 3 deviation from the Experiment 1 deviation, that is, $(\mathcal{D}\Delta\text{WTP}_{1}-\mathcal{D}\Delta\text{WTP}_{3})$.\footnote{Note that $\mathcal{D}\Delta\text{WTP}_{1}-\mathcal{D}\Delta\text{WTP}_{3}=(\Delta\text{WTP}_{1} - \Delta\text{WTP}_{1}^{Theory})-(\Delta\text{WTP}_{3} - \Delta\text{WTP}_{3}^{Theory})=\Delta\text{WTP}_{1} - \Delta\text{WTP}_{3}$ because $\Delta\text{WTP}_{1}^{Theory} = \Delta\text{WTP}_{3}^{Theory}$ for the corresponding admission probabilities.} This difference captures the residual deviation associated with mistakes in mapping others' learning decisions into one's own admission chances.\footnote{Because Experiment 3 is conducted after Experiment 1, differences between matched responses may also reflect learning, fatigue, or other carryover effects. These factors may affect the magnitude of the residual deviations, but are unlikely to explain the observed treatment differences.}

\begin{figure}[htbp]
 \centering
 \includegraphics[width=0.8\linewidth]{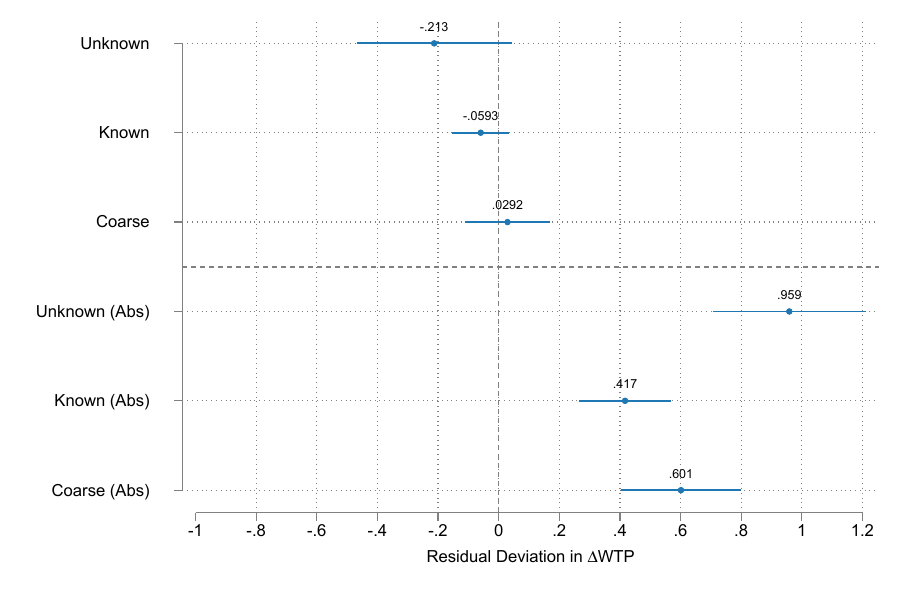}
 \caption{Residual Deviation in $\Delta WTP$ in Experiment 1}
 \label{fig:Desect_Exp1_abs}
  \begin{tabnotes}
  The figure shows the residual deviation after subtracting the Experiment 3 deviation from the Experiment 1 deviation. ``Unknown,'' ``Known,'' and ``Coarse'' represent $\mathcal{D}\Delta\text{WTP}_{1}-\mathcal{D}\Delta\text{WTP}_{3}$. ``Unknown (Abs),'' ``Known (Abs),'' and ``Coarse (Abs)'' represent $|\mathcal{D}\Delta\text{WTP}_{1}-\mathcal{D}\Delta\text{WTP}_{3}|$ in the corresponding treatments. Average levels are presented with 95\% confidence intervals, and standard errors are clustered at the subject level.
  \end{tabnotes}
\end{figure}

Figure \ref{fig:Desect_Exp1_abs} presents the average levels of $(\mathcal{D}\Delta\text{WTP}_{1}-\mathcal{D}\Delta\text{WTP}_{3})$ and $|\mathcal{D}\Delta\text{WTP}_{1}-\mathcal{D}\Delta\text{WTP}_{3}|$ across the three treatments. Although the signed averages of $(\mathcal{D}\Delta\text{WTP}_{1}-\mathcal{D}\Delta\text{WTP}_{3})$ are close to zero, especially \textit{Known} and \textit{Coarse}, the absolute magnitudes indicate that substantial residual deviations remain even after accounting for deviations in Experiment 3. These residual deviations are equivalent to 94\% of the baseline Experiment 1 deviation in Unknown, 74.2\% in Known, and 77.2\% in Coarse.\footnote{These proportions are calculated by dividing $|\mathcal{D}\Delta\text{WTP}_{1}-\mathcal{D}\Delta\text{WTP}_{3}|$ by $|\mathcal{D}\Delta\text{WTP}_{1}|$ in the corresponding treatments.}

Regression (3) in Table \ref{tab:deviation in DWTP_exp1vs3} confirms that the residual deviation is significantly larger in \textit{Unknown}, suggesting a poorer understanding of the strategic complementarity in learning. Although the residual deviation is larger in \textit{Coarse} than in \textit{Known}, this difference is not statistically significant ($p=0.135$).

The results on deviations in $\Delta\text{WTP}$ are summarized as follows.

\medskip{}

\noindent \textbf{Result 4 (Deviation in $\Delta\text{WTP}$)}

\noindent \textit{(1) Across all three treatments, subjects respond insufficiently to changes in others' learning decisions in Experiment 1. The absolute deviations are smallest in \textit{Known}, followed by \textit{Coarse}, and largest in \textit{Unknown}; the \textit{Known}--\textit{Unknown} difference is significant at 5\%, and the other pairwise differences at 10\%.}

\noindent \textit{(2) In \textit{Unknown}, these deviations are driven primarily by mistakes in mapping others' learning decisions into one's own admission chances.}

\noindent \textit{(3) In \textit{Known} and \textit{Coarse}, these deviations are also attributed to suboptimal responses to directly provided admission chances.}

\medskip{}

These results suggest that subjects in \textit{Unknown} face greater difficulty understanding the strategic environment, likely because they must reason under greater uncertainty about their priorities. Greater priority transparency in \textit{Known} and \textit{Coarse} reduces this strategic difficulty, but deviations in learning remain and can be partly explained by suboptimal responses to directly provided admission chances. This highlights the importance of addressing not only strategic mistakes, but also subjects' ability to use admission-chance information when designing policies to improve learning in such environments.

\subsection{School Decisions}

Next, we provide a brief overview of subjects' choices between Schools A and B in the admission procedure. We find that school decisions are largely consistent with Hypothesis 5: subjects choose School B after observing a realized School-B payoff of 40 and choose School A otherwise. There is no clear evidence of learning effects across rounds or differences across treatments. Among the 27 incorrect decisions, only one involves the subject choosing the school with the lower payoff after learning. In the remaining cases, subjects did not acquire information but chose School B over School A.

As noted above, school decisions are not the primary focus of the paper, and we intentionally simplify this part of the design. Nonetheless, these results validate the assumption that students prefer School A to School B when they do not acquire information, and they also provide evidence of subject attention and rationality during the experiment. The results on school choices are summarized as follows.

\medskip{}

\noindent \textbf{Result 5 (School Choices)}

\noindent \textit{Overall, 93.63\% of school decisions are correct, with no significant differences across treatments or rounds.}

\subsection{Beliefs}
Lastly, we examine subjects’ beliefs about other students’ learning decisions in Experiment 2. All students in \textit{Known} and Student 1/2 in \textit{Coarse} predict the learning decisions of each of the other two students, so a guess lies in \(\{0,1\}^2\). All students in \textit{Unknown} and Student 3 in \textit{Coarse} predict the number of other students who learn, so a guess lies in \(\{0,1,2\}\). To compare across treatments and roles, we convert guesses in the first case to the corresponding number in \(\{0,1,2\}\), and score a guess as correct if the predicted number of other students who learn is correct. The share of correct guesses is similar across treatments: 40\% in \textit{Known}, 44\% in \textit{Coarse}, and 41\% in \textit{Unknown}.

\medskip{}

\noindent \textbf{Result 6 (Beliefs)}

\noindent \textit{The proportion of correct beliefs about other students' learning decisions is similar across \textit{Known}, \textit{Coarse}, and \textit{Unknown}.}

% \medskip{}

\section{Conclusion}\label{sec:Conclusion}

This paper studies how priority transparency affects students' incentives to acquire costly information about their own preferences in school choice and college admissions. Our theoretical analysis shows that greater priority transparency does not always improve welfare. In some cases, partial disclosure can sustain stronger learning incentives by pooling students with different priority positions. Thus, although transparency is privately valuable, full disclosure need not be socially optimal.

The experimental results point in a different direction. In the laboratory, full knowledge of priorities yields the highest welfare, followed by partial knowledge, while no knowledge performs worst. This ranking contrasts with the theoretical prediction that partial disclosure can dominate full disclosure. The reason is behavioral: subjects make fewer mistakes when their priorities are disclosed more precisely. They respond more appropriately to their disclosed priorities, are more likely to treat learning decisions as strategic complements. By contrast, subjects in the no-knowledge treatment face greater difficulty understanding the strategic environment and mapping others' learning decisions into their own admission chances.

The experiment also shows that difficulties in learning decisions are not purely strategic. Even when priority transparency reduces strategic complexity, deviations in learning remain and are partly explained by suboptimal responses to directly provided admission chances. This suggests that disclosure is not sufficient on its own. The way admission chances are presented, and whether students are able to use them when deciding how much information to acquire, are themselves important policy concerns.

Our findings have two main implications. First, priority transparency can improve welfare by helping students make better learning decisions. This supports policies that disclose priorities or admission chances early enough to guide information acquisition. Second, designers should not assume that students will automatically use disclosed priorities or admission chances optimally. Transparency should be accompanied by clear guidance and decision support. More broadly, the paper highlights the importance of incorporating both information-acquisition incentives and behavioral responses when designing transparent matching mechanisms.

\section*{Acknowledgments}

We thank David Byrne, Yan Chen, Rustamdjan Hakimov, Dorothea Kübler, Simon Loertscher, Vincent Meisner, Andy Skrzypacz, Ao Wang, and Tom Wilkening for helpful comments and suggestions. We also thank participants at the Econometric Society Australasian Meeting, the Australian Conference of Economists, the ARC Behavioural Market Design Workshop, the Annual Australia New Zealand Workshop in Experimental Economics, the Greater Bay Area Market Design Workshop, the Deakin Economic Theory Workshop, the Conference on Economic Design, and the Matching Market Design: Strategy-Proofness and Beyond Workshop, as well as seminar participants at the University of Michigan, the Ohio State University, the WZB Berlin Social Science Center, Boston College, Brown University, the University of Sydney, the University of Queensland, the Chinese University of Hong Kong, Shenzhen, and the University of International Business and Economics, Beijing. We gratefully acknowledge financial support from MatchLab in the Faculty of Business and Economics at the University of Melbourne (Artemov and Pan), and from the Australian Research Council through Discovery Project DP160101350 (Artemov) and Discovery Early Career Researcher Award DE250100969 (Pan).

\bibliography{schoolChoice}
\bibliographystyle{ecta}

\newpage
\appendix
\def\thesection{Appendix \Alph{section}}
\section{A Continuum Model}\label{sec:general_model}
\def\Dpu{\Delta}

An important question is whether the key intuitions and theoretical predictions from our experimental setting continue to hold in a more general environment with more students and a broader range of parameters. In this section, we replicate the theoretical results from the main text in a large-market setting.

There is a unit mass of students, indexed by $i \in [0,1]$, who are assigned to one of two schools, A and B, through a sequential Random Serial Dictatorship (RSD) mechanism. For each student, the value of School B is independently drawn and initially unknown but can be learned at a cost. The timing is as follows: students first decide whether to pay the cost and learn the value of School B, and then RSD is run. Because the mechanism is strategy-proof, we assume that students choose between two schools according to their true preferences: when their turn to select a school arrives, they select the available school with the higher expected utility. Thus, student $i$'s strategy concerns only information acquisition. Let $e(i)=0$ if she decides not to learn the value of School B, and $e(i)=1$ if she does.

The value of School A is normalized to 1. Student $i$'s value of School B is $u_i(B)=u_B^H$ with probability $p$ and $u_i(B)=u_B^L$ with probability $(1-p)$, where $u_B^L<1<u_B^H$ and the expected value $u_B=pu_B^H+(1-p)u_B^L$
is normalized to $0$. Thus, before learning School B's realized value, every student prefers School A. After learning, student $i$ prefers School A to School B with probability $1-p$, when $u_i(B)=u_B^L$, and prefers School B to School A with probability $p$, when $u_i(B)=u_B^H$. The value of an outside option is below $u_B^L$, so both schools are acceptable to all students. Each student's cost of information acquisition is a common constant $c$.\footnote{The common cost assumption delivers a clean result, but a cost distribution with costs sufficiently concentrated in the $(\alpha\Dpu,\alpha p u_B^H)$ interval would lead to similar results. See the proof of Theorem~\ref{prop:Wcompare}, equations (\ref{Wcompare:ck}-\ref{Wcompare:uk}) for details on this interval.}

Each student independently draws a priority $s_i$ from a uniform distribution on $[0,1]$, with lower draws indicating higher priority at both schools.\footnote{With a uniform distribution, the assumption that lower draws imply higher priorities allows us to interpret priorities as admission chances in what follows.} Under sequential RSD, schools are non-strategic, and students choose sequentially from the available schools in increasing order of $s_i$. Because only the ordinal ranking of priorities matters, the uniform distribution is without loss of generality.

The capacity of School B is 1, so it can accommodate all students. The capacity of School A is $q^H$ with probability $\alpha$ and $q^L$ with probability $(1-\alpha)$, where $0<q^L<q^H<1-p$. Its expected capacity is therefore $\bar q = \alpha q^H + (1-\alpha) q^L$. The condition $q^H < 1 - p$ guarantees that School A remains oversubscribed even when capacity is $q^H$ and even if all students acquire information.

Capacity uncertainty may reflect genuine uncertainty about capacity, unmodeled uncertainty about the mass of high-priority students, or the finite-market uncertainty faced by students. Without such uncertainty, students in some treatments could perfectly predict her assignment, leaving no ``marginal'' students---a feature that is unlikely to be realistic in most settings.

\subsection{Preliminaries}

Because School B can accommodate all students, it is always available when a student chooses. We therefore focus on student $i$'s admission chance at School A, denoted by $a_i$: the probability that School A remains available when it is her turn to choose.

\begin{lm}\label{LemDiffUtil}
Suppose that student $i$'s admission chance at School A is $a_i$. Then the expected gain from acquiring information, excluding the information cost, is
\begin{align}
\mathbb{E}\left[u_i\big(e(i)=1\big)-u_i\big(e(i)=0\big)\right]
= a_i p(u_B^H-1),
\end{align}
where the expectation is taken over the realization of student $i$'s value of School B.
\end{lm}

\begin{proof}[Proof of Lemma \ref{LemDiffUtil}]
When student $i$ is informed and prefers School B, she chooses B and is assigned to it with probability 1. When she prefers School A, she chooses A and is assigned to it with probability $a_i$, deriving utility 1; otherwise, she is assigned to B. When she is uninformed, she chooses A whenever it is available. She is assigned to A with probability $a_i$, deriving utility 1, and to B with probability $1-a_i$, deriving expected utility 0. Thus,
\begin{align*}
\mathbb{E}\left[u_i\big(e(i)=1\big)-u_i\big(e(i)=0\big)\right]
&= \big[p u_B^H+(1-p)(a_i \times 1+(1-a_i)u_B^L)\big]-\big[a_i\times1 + (1-a_i)\times 0\big] \\
&= (1-a_i)\big[p u_B^H+(1-p)u_B^L\big]
+a_i\big(p u_B^H-p\big) \\
&= a_i p(u_B^H-1),
\end{align*}
where the last equality follows from the normalization
$p u_B^H+(1-p)u_B^L=0$.
\end{proof}

The following lemma facilitates equilibrium analysis and also mirrors our discussion of complementarity for the three-student model in the main text. 

\begin{lm}\label{lm:complementarity}
Fix a capacity realization. Student $i$'s admission chance at School A, $a_i$, is weakly increasing in the mass of higher-priority students who acquire information and is unaffected by the learning decisions of lower-priority students.
\end{lm}
\begin{proof}
	For a given capacity realization, School A is available to student $i$ if and only if, when her turn arrives, the mass of higher-priority students who have chosen School A is below its capacity. When School A is available, an uninformed student always chooses School A, whereas an informed student chooses School A only when her value of School B is low, which occurs with probability $1-p$. Thus, when more higher-priority students acquire information, the mass of higher-priority students choosing School A weakly decreases. Student $i$'s admission chance at School A is therefore weakly increasing.

    Students with lower priorities than $i$ choose only after her and hence do not affect whether School A is available when her turn arrives. Therefore, their learning decisions do not affect $a_i$. By Lemma \ref{LemDiffUtil}, the gain from information acquisition, $a_i p(u_B^H-1)$, inherits the same monotonicity.
    \end{proof}

In the equilibrium analysis below, let $r$ denote the admission cutoff for School A: students with $s_i<r$ find School A available when their turn arrives, whereas those with $s_i>r$ do not. Let $r_z^L$ and $r_z^H$ denote the cutoffs under capacity realizations $q^L$ and $q^H$, respectively, where $z\in\{u,k,c\}$ indexes Treatments \textit{Unknown}, \textit{Known}, and \textit{Coarse}.

Under the uniform distribution, a student's probability of having priority at least as high as a cutoff $r$ is $r$. Hence, without knowing her priority, she faces an admission chance of $r$ at School A. We refer to students for whom School A is feasible under both capacity realizations as the ``top band,'' those for whom it is feasible only when capacity is high as the ``middle band,'' and those for whom it is never feasible as the ``bottom band.'' The admission cutoffs, and therefore this partition, are endogenous to the equilibrium strategy profile.

To simplify the notation that follows, let $\Dpu=p(u_B^H-1)$. By Lemma \ref{LemDiffUtil}, student $i$'s expected gain from acquiring information is $a_i\Dpu$. Thus, a top-band student, for whom $a_i=1$, gains $\Dpu$ from acquiring information. A middle-band student gains $\alpha\Dpu$, since School A is available only when capacity is high. Throughout, we assume that
\begin{align}
\alpha\Dpu<c<\Dpu. \tag{$\star$}\label{eq:star}
\end{align}
As discussed in \ref{sec:Known}, under (\ref{eq:star}), a student who knows she is in the top band acquires information, whereas one who knows she is in the middle band does not. We focus on this range because Treatments \textit{Known} and \textit{Coarse} yield different predictions only for these costs. When $c<\alpha\Dpu$, both top- and middle-band students acquire information under both treatments; when $c>\Dpu$, no student acquires information under either treatment.

\subsection{Unknown Priorities}

In this section, we derive results for the treatment in which students do not know their priorities when making their learning decisions. The following proposition characterizes the mass of students who acquire information in the maximum-information equilibrium. As discussed in the main text, our theoretical analysis focuses on the maximum-information equilibrium because it delivers the highest total welfare: information acquisition generates a positive externality and thus increases welfare whenever it is privately optimal. Our experimental design allows us to control for equilibrium selection in the data analysis.

\begin{prop}\label{ProUnkn}
	In the maximum-information equilibrium under Treatment \textit{Unknown}, every student acquires information if $c \leq \Dpu\,\bar q/(1-p)$ and no student acquires information otherwise.
\end{prop}

\begin{proof}[Proof of Proposition \ref{ProUnkn}]

Let $\gamma$ denote the mass of students who acquire information. In the low-capacity state $q^L$, the mass of students choosing School A is
\begin{align}
r_u^L\big[\gamma(1-p)+(1-\gamma)\big]=r_u^L(1-\gamma p).
\end{align}
Market clearing therefore requires 
\begin{align}
r_u^L(1-\gamma p)=q^L,
\end{align}
or equivalently,
\begin{align}
r_u^L=\frac{q^L}{1-\gamma p}.
\end{align}
Similarly, in the high-capacity state $q^H$,
\begin{align}
r_u^H=\frac{q^H}{1-\gamma p}.
\end{align}
By assumption, $q^H<1-p$, so $r_u^L,r_u^H\in(0,1)$ for every $\gamma\in[0,1]$.

By Lemma \ref{LemDiffUtil}, the expected gain from acquiring information for every student is
	\begin{align}
		\alpha\frac{q^H}{1-\gamma p}\Dpu + (1-\alpha)\frac{q^L}{1-\gamma p}\Dpu = \frac{\bar q}{1-\gamma p}\Dpu.
	\end{align}
This gain is increasing in $\gamma$ and is therefore maximized at $\gamma=1$. Hence, if
\begin{align}
c\leq \frac{\bar q}{1-p}\Dpu,
\end{align}
then acquiring information is optimal when every student acquires information, so $\gamma=1$ is the maximum-information equilibrium. Conversely, if
\begin{align}
c>\frac{\bar q}{1-p}\Dpu,
\end{align}
then the gain from acquiring information is below $c$ for every $\gamma\in[0,1]$. Hence, no student acquires information.
\end{proof}

\subsection{Known Priorities}\label{sec:Known}

In this section, we derive results for the treatment in which students observe their precise priority $s_i$ before making their learning decisions. Their strategies may therefore depend on $s_i$.

\begin{prop}\label{prop:eq_k}
	Under Treatment \textit{Known}, there is a unique equilibrium in which students with $s_i \leq r_k^L$ acquire information, whereas those with $s_i > r_k^L$ do not. The cutoffs under $q^L$ and $q^H$ are
	\begin{align}
		r_k^L = \frac{q^L}{1-p},\label{eq:kn:rL}\\
		r_k^H = q^H - q^L + \frac{q^L}{1-p}.
	\end{align}
\end{prop}

\begin{proof}
In \textit{Known}, a top-band student with $s_i\leq r_k^L$ knows that her admission chance is $a_i=1$. By Lemma \ref{LemDiffUtil}, her expected gain from acquiring information is $\Dpu$. Since $\Dpu>c$ under Assumption (\ref{eq:star}), she acquires information.

A middle-band student with $s_i\in(r_k^L,r_k^H]$ knows that School A is feasible only when capacity is high. Her admission chance is therefore 1 when capacity is $q^H$ and 0 when capacity is $q^L$. Her expected gain from acquiring information is
$\alpha\Dpu+(1-\alpha)\times 0=\alpha\Dpu<c,$
by Assumption (\ref{eq:star}). Hence, she does not acquire information.

A bottom-band student with $s_i>r_k^H$ knows that School A is never feasible. Her gain from acquiring information is therefore zero, and she does not acquire information. 

We find the cutoffs next. Note that, because students do not observe capacity realizations, they cannot condition their learning decisions on it.

Under $q^L$, School A is feasible only for top-band students, $s_i\leq r_k^L$. These students acquire information, and a fraction $1-p$ choose School A. Thus, the cutoff $r_k^L$ that clears the market satisfies
\begin{align}
    r_k^L(1-p)=q^L,
\end{align}
which gives
\begin{align}
    r_k^L=\frac{q^L}{1-p}.
\end{align}

Under $q^H$, top-band students occupy $q^L$ seats at School A, leaving $q^H-q^L$ seats for middle-band students with $s_i\in(r_k^L,r_k^H]$. These students do not acquire information and therefore all choose School A. Thus, market clearing requires
\begin{align}
r_k^H-r_k^L=q^H-q^L,
\end{align}
which gives
\begin{align}
r_k^H=q^H-q^L+\frac{q^L}{1-p}.
\end{align}
Under Assumption (\ref{eq:star}), all relevant inequalities are strict, so the equilibrium is unique.
\end{proof}

\subsection{Coarse Priorities}

Treatment \textit{Coarse} informs each student whether her priority lies in $[0,\tilde{r}_c^H]$ or in $(\tilde{r}_c^H,1]$, where $\tilde{r}_c^H = q^H/(1-p)$ denotes the equilibrium admission cutoff when capacity is $q^H$ and $\gamma=1$. Thus, when making her learning decision, each student knows whether she belongs to the top-or-middle pool or to the bottom band. A pooled student does not know whether her priority places her in the top band or the middle band and therefore evaluates her admission chance at School A by averaging over these possibilities. This is the continuum counterpart of the top-two-versus-third coarsening used in the experiment.

\begin{prop}\label{prop:eq_c}
	Under Treatment \textit{Coarse}, if $c \leq \left(\alpha + (1-\alpha)\frac{q^L}{q^H}\right)\Dpu$, then the maximum-information \footnote{There are multiple equilibria when the cost is in the range $(1-p)\frac{\bar q}{q^H}\Dpu\leq c\leq\frac{\bar q}{q^H}\Dpu$.} is such that students with $s_i \leq r_c^H$ acquire information, whereas those with $s_i > r_c^H$ do not. The cutoffs under $q^L$ and $q^H$ are
	\begin{align}
		r_c^L = \frac{q^L}{1-p}, \qquad r_c^H = \frac{q^H}{1-p}.\label{eq:eq_c_cutoffs}
	\end{align}
	If $c > \left(\alpha + (1-\alpha)\frac{q^L}{q^H}\right)\Dpu$, then no student acquires information.
\end{prop}

\begin{proof}
First, students with $s_i>\tilde{r}_c^H$ are assigned to School B under either capacity realization, because the mass of pooled students choosing School A is at least $\tilde{r}_c^H(1-p)=q^H$. Their gain from acquiring information is therefore zero, so they do not acquire information.

Suppose that a fraction $\gamma$ of students with $s_i\leq\tilde{r}_c^H$ acquires information. 
Market clearing requires
\begin{align}
    r_c^L\big[\gamma(1-p)+(1-\gamma)\big]  &= q^L, \\
(r_c^H-r_c^L)\big[\gamma(1-p)+(1-\gamma)\big] &= q^H-q^L.
\end{align}
Hence,
\begin{align}
    r_c^L=\frac{q^L}{1-\gamma p},
    \qquad
    r_c^H=\frac{q^H}{1-\gamma p}.
\end{align}

Consider a student in the pooled group with $s_i\in[0,\tilde{r_c^H}]$. Conditional on belonging to this group, her priority is uniformly distributed on $[0,\tilde{r_c^H}]$. When capacity is $q^H$, School A is available to her with probability $(q^H/(1-\gamma p)/(q^H/(1-p) = (1-p)/(1-\gamma p)$. When capacity is $q^L$, it is available with probability
${r_c^L}/{\tilde{r}_c^H}={q^L(1-p)}/{q^H(1-\gamma p)}$.
By Lemma \ref{LemDiffUtil}, her expected gain from acquiring information is therefore
\begin{align}\label{eq:learningBenefit}
\frac{1-p}{1-\gamma p}\left(\alpha+(1-\alpha)\frac{q^L}{q^H}\right)\Dpu.
\end{align}

Suppose that
\begin{align}
c<\left(\alpha+(1-\alpha)\frac{q^L}{q^H}\right)\Dpu.
\end{align}
We show that $\gamma=1$ is an equilibrium by observing that the benefit of learning is then exactly $\left(\alpha+(1-\alpha)\frac{q^L}{q^H}\right)\Dpu$, which is above the cost. No student has incentive to deviate to not learning.

Suppose that  
\begin{align} c > \left(\alpha+(1-\alpha)\frac{q^L}{q^H}\right)\Dpu.
\end{align}
As the learning benefit in equation (\ref{eq:learningBenefit}) is at most $\left(\alpha+(1-\alpha)\frac{q^L}{q^H}\right)\Dpu,$ the cost is above the learning benefit and $\gamma = 0$. No student has incentive to deviate to learning.
\end{proof}

\subsection{Welfare Comparison}

In this section, we compare welfare across the three treatments. To further simplify the notation, we index students by their priorities and write $i = s_i$. 

Student $i$'s welfare can be written as the utility she would obtain if uninformed plus the gain from information acquisition given by Lemma \ref{LemDiffUtil}. Aggregating across students, the welfare is
\begin{align}
    W=\int_0^1 a_i\,di+\int_{\{i\,:\,e(i)=1\}}\left(a_i\Dpu-c\right)\,di,
\label{eq:welfare_id}
\end{align}
where $a_i$ is student $i$'s equilibrium admission chance at School A. We use $W_u$, $W_k$, and $W_c$ to denote welfare under Treatments \textit{Unknown}, \textit{Known}, and \textit{Coarse}, respectively.

\begin{thm}\label{prop:Wcompare}
Let
\begin{align*}
	\underline c = \min\left\{\frac{(q^H-q^L)\alpha p u_B^H}{1-p-q^L},\ \frac{\bar q}{1-p}\Dpu\right\},
	\qquad
	\overline c = \min\left\{\alpha p u_B^H,\ \frac{\bar q}{q^H}\Dpu\right\}.
\end{align*}
Then $\underline c<\overline c<\Dpu$, and, as the information cost $c$ rises within the range $\alpha\Dpu<c<\Dpu$ in (\ref{eq:star}), the welfare ranking is as follows.\footnote{The first interval, $\alpha\Dpu<c<\underline c$, in which \textit{Unknown} dominates \textit{Known}, may be empty for some parameter values. The other two intervals are always nonempty, as follows from $\underline c<\overline c<\Dpu$. In the final interval, $\overline c<c<\Dpu$, \textit{Coarse} and \textit{Unknown} coincide once neither treatment induces information acquisition.}
\begin{itemize}
	\item \textbf{\textit{Coarse} $\succ$ \textit{Unknown} $\succ$ \textit{Known}} \quad for $\alpha\Dpu<c<\underline c$;
	\item \textbf{\textit{Coarse} $\succ$ \textit{Known} $\succ$ \textit{Unknown}} \quad for $\underline c<c<\overline c$;
	\item \textbf{\textit{Known} $\succ$ \textit{Coarse} $\succeq$ \textit{Unknown}} \quad for $\overline c<c<\Dpu$.
\end{itemize}

\end{thm}

\begin{proof}
We start with \textit{Known}. By Proposition \ref{prop:eq_k}, for every cost satisfying (\ref{eq:star}), top-band students, $[0,r_k^L]$, acquire information, whereas middle- and bottom-band students do not. The cutoffs are $r_k^L=q^L/(1-p)$ and $r_k^H=r_k^L+(q^H-q^L)$.

In the first term of \eqref{eq:welfare_id}, the admission chance is 1 on $[0,r_k^L]$, $\alpha$ on $(r_k^L,r_k^H]$, and 0 on $(r_k^H,1]$. Thus,
\[
\int_0^1 a_i\,di
=
r_k^L+\alpha(r_k^H-r_k^L)
=
\frac{q^L}{1-p}+\alpha(q^H-q^L).
\]
Only the top band acquires information, and these students have admission chance 1. The second term in \eqref{eq:welfare_id} is therefore
\[
\int_0^{r_k^L}(\Dpu-c)\,di
=
r_k^L(\Dpu-c).
\]
Thus,
\begin{align}
W_k
&= \frac{q^L}{1-p}+\alpha(q^H-q^L)
+ \frac{q^L}{1-p}(\Dpu-c) \\
&= \frac{q^L}{1-p}(1+\Dpu-c)
+ \alpha(q^H-q^L) \\
&= \bar q+\frac{q^L}{1-p}(p+\Dpu-c).
\end{align}

Whenever students acquire information under both \textit{Coarse} and \textit{Unknown}, the cutoffs are the same:
\[
r_{c,u}^L=\frac{q^L}{1-p},
\qquad
r_{c,u}^H=\frac{q^H}{1-p}.
\]
Thus, under both treatments,
\[
\int_0^1 a_i\,di
=
r_{c,u}^L+\alpha(r_{c,u}^H-r_{c,u}^L)
=
\frac{\bar q}{1-p}.
\]

Under \textit{Coarse}, only students with $i\leq r_{c,u}^H$ acquire information. Splitting the second term in \eqref{eq:welfare_id} by the top and middle bands---with $a_i=1$ on $[0,r_{c,u}^L]$, $a_i=\alpha$ on $(r_{c,u}^L,r_{c,u}^H]$, and $a_i=0$ on $(r_{c,u}^H,1]$---gives
\begin{align*}
W_c
&=
\frac{\bar q}{1-p}
+r_{c,u}^L(\Dpu-c)
+(r_{c,u}^H-r_{c,u}^L)(\alpha\Dpu-c) \\
&=
\frac{\bar q}{1-p}(1+\Dpu)
-\frac{q^H}{1-p}c.
\end{align*}
Under \textit{Unknown}, all students acquire information. Hence,
\begin{align*}
W_u
&=W_c-c(1-r_{c,u}^H) \\
&=\frac{\bar q}{1-p}(1+\Dpu)-c.
\end{align*}
The additional term $-c(1-r_{c,u}^H)$ is the information cost incurred by bottom-band students, who acquire information only under \textit{Unknown}.

Comparing $W_c$ and $W_u$ with $W_k$, and using $\Dpu+p=pu_B^H$, gives
\begin{align}
W_c-W_k
&=
\frac{q^H-q^L}{1-p}\bigl(\alpha p u_B^H-c\bigr),\label{Wcompare:ck} \\
W_u-W_k
&=
\frac{q^H-q^L}{1-p}\alpha p u_B^H
-c+\frac{q^L}{1-p}c.\label{Wcompare:uk}
\end{align}
These expressions apply when the relevant treatments induce information acquisition. Both differences decrease in $c$, and therefore
\begin{align*}
W_c>W_k
&\iff
c<\alpha p u_B^H, \\
W_u>W_k
&\iff
c<
\frac{(q^H-q^L)\alpha p u_B^H}{1-p-q^L},
\end{align*}
as long as the stated cost cutoffs are within the validity region---when students in the treatments being compared learn.

The treatments differ in the costs at which information acquisition ceases. Under \textit{Unknown}, students acquire information if and only if $c<\frac{\bar q}{1-p}\Dpu$, whereas under \textit{Coarse}, the pooled group acquires information if and only if $c<\frac{\bar q}{q^H}\Dpu$ (Propositions \ref{ProUnkn} and \ref{prop:eq_c}). Since $q^H<1-p$,
$\frac{\bar q}{1-p}<\frac{\bar q}{q^H}$,
so information acquisition ceases earlier under \textit{Unknown}.

Once a treatment no longer induces information acquisition, its welfare is $\bar q$. Since $\bar q<W_k$ for any costs in (\ref{eq:star}), combining the welfare comparisons above with the relevant learning thresholds yields
\begin{align*}
W_c>W_k
&\iff
c<\overline c
=
\min\left\{
\alpha p u_B^H,
\frac{\bar q}{q^H}\Dpu
\right\}, \\
W_u>W_k
&\iff
c<\underline c
=
\min\left\{
\frac{(q^H-q^L)\alpha p u_B^H}{1-p-q^L},
\frac{\bar q}{1-p}\Dpu
\right\}.
\end{align*}
Finally, $p<1-q^H$ implies
\[
\frac{(q^H-q^L)\alpha p u_B^H}{1-p-q^L}
<
\alpha p u_B^H.
\]
Together with $\frac{\bar q}{1-p}<\frac{\bar q}{q^H}$,
this yields $\underline c<\overline c<\Dpu$.

For the comparison between \textit{Coarse} and \textit{Unknown}, whenever both induce information acquisition, $W_c>W_u$ because \textit{Unknown} additionally induces bottom-band students to acquire information. \textit{Unknown} ceases to induce learning before \textit{Coarse}, after which $W_u=\bar q<W_c$ until \textit{Coarse} also ceases to induce learning. Hence, \textit{Coarse} strictly dominates \textit{Unknown} whenever
$c<\frac{\bar q}{q^H}\Dpu$,
and the two treatments coincide thereafter.
\end{proof}

The intuition is as follows. When information costs are low, but still within the range of (\ref{eq:star}), only the top band acquires information under \textit{Known}. Under \textit{Coarse}, both the top and middle bands acquire information, while under \textit{Unknown}, all students do. Learning is costly, and some of it is wasteful under \textit{Coarse} and \textit{Unknown}: under both treatments, learning by the middle band is wasted in the low-capacity state, and under \textit{Unknown}, learning by the bottom band is wasted in both capacity states. However, this waste is outweighed by the positive externality generated by learning among middle-band students.

As the learning cost rises, the cost of wasteful learning increases, while the positive externality remains constant. Eventually, this waste exceeds the externality generated by pooled students' learning. The other regime switch, as costs rise, occurs when learning ceases, first under \textit{Unknown} and then under \textit{Coarse}. Once learning stops, the positive externality that favored \textit{Coarse} and \textit{Unknown} disappears, ensuring that \textit{Known} dominates them.

For each comparison, the relevant threshold is determined by whichever occurs first: the point at which wasteful learning outweighs its externality, or the point at which learning ceases. This is why $\underline c$ and $\overline c$ are each defined as the minimum of two cost thresholds.

\section{Recombinant Estimation}\label{sec:Recombinant}
We use recombinant estimation to construct the welfare measures used to compare the three treatments. Specifically, we combine subjects into different groups of three and use their elicited strategies to compute the payoffs that would have occurred under such groupings.

\subsection{Experiment 2}
In Experiment 2, subjects know the costs of the other two students but do not receive any feedback about their group members' actions. Their decisions are therefore not affected by the composition of the group. We form all possible groups of three subjects, use their reported WTP at different costs to determine their learning decisions, and then calculate the resulting payoffs.

Under \textit{Unknown}, there are 63 subjects who all play the identical role, Student 1/2/3. This gives $63 \times 62 \times 61 = 238{,}266$ possible groupings. Under \textit{Known}, there are 63 subjects, with 21 subjects in each role: Student 1, Student 2, and Student 3. Since each subject plays a specific role, there are $21^3 = 9{,}261$ possible groupings. Under \textit{Coarse}, subjects playing the role of Student 1/2 can be recombined arbitrarily with one another, but not with subjects playing Student 3. With 54 subjects in total, this gives $36 \times 35 \times 18 = 22{,}680$ possible groupings.

Each group of three subjects faces 27 combinations of learning costs from $\{1.8, 2.6, 3\}$. To determine a subject's learning decision for a given cost triple, we compare her reported WTP for the relevant cost environment with her own cost in that triple. For example, suppose the cost triple is $(2.6,1.8,3)$ and the first subject's WTP when the other two students' costs are $(1.8,3)$ is 2. We then conclude that she does not acquire information, since $2<2.6$. Given the learning decisions of all three subjects, we calculate individual payoffs and then aggregate payoffs across all groups for each cost triple. Note that $(2.6,1.8,3)$ and $(1.8,2.6,3)$ correspond to the same cost triple under \textit{Unknown} and \textit{Coarse}, but not under \textit{Known}.

For each cost combination, we separately test the null hypotheses that the mean welfare differences are zero for all pairwise treatment comparisons. The reported $p$-values are from two-sided Wald \(z\)-tests using H\'{a}jek-projection variance estimators, which account for dependence among recombinant groups that share subjects, following the approach of \cite{abrevaya2008recombinant}.

\subsection{Experiment 1}
A particular challenge in Experiment 2 is the multiplicity of equilibria under \textit{Unknown} and \textit{Coarse}.\footnote{Although multiple equilibria arise under \textit{Known} for some cost combinations, those combinations are not offered to subjects in Experiment 2.} Experiment 1 avoids this difficulty by informing subjects of the learning decisions, rather than the learning costs, of the other students.

To present the welfare results of Experiment 1 in a way that is comparable to Experiment 2, where decisions are based on costs rather than actions, we map costs into actions. Because of equilibrium multiplicity, and because the same subject may submit different WTP reports in strategically equivalent situations, this mapping is not unique.\footnote{For example, a Student 1 subject under \textit{Known} may report different WTP values depending on the actions of others, even though Students 2 and 3 are not relevant to her and her payoff is therefore independent of their actions.} We therefore calculate two extremes, corresponding to the maximum-information and minimum-information beliefs about other students' learning decisions. The mapping for any intermediate belief lies between these two extremes.

For example, consider \textit{Unknown}. Suppose a subject reports $\text{WTP}(0)=2$ when neither of the other two students, played by robots, learns; $\text{WTP}(1)=3$ when one other student learns; and $\text{WTP}(2)=4$ when both other students learn. Suppose further that the subject's own cost is 1 in the cost triple. To determine which WTP value applies, we infer the other students' learning decisions from their costs. If the other two students' costs are $(1,1)$, it is rational for both of them to learn. We therefore use $\text{WTP}(2)=4$ for the subject and conclude that she learns. If their costs are $(3,3)$, it is rational for neither of them to learn, so we use $\text{WTP}(0)=2$. Finally, if their costs are $(2.6,2.6)$, there are two equilibria: both other students learn, or neither other student learns. The maximum-information calculation uses $\text{WTP}(2)=4$, whereas the minimum-information calculation uses $\text{WTP}(0)=2$.

For each triple of subjects, each cost combination $(x,y,z)$, and each mapping rule---maximum-information or minimum-information---there is a unique WTP triple, $\big(\text{WTP}(y,z),\allowbreak \text{WTP}(x,z),\allowbreak \text{WTP}(x,y)\big)$. We compare this WTP triple with the cost triple to determine subjects' learning decisions. For example, if $\text{WTP}(y,z)<x$, $\text{WTP}(x,z)<y$, and $\text{WTP}(x,y)>z$, then subjects 1 and 2 do not learn, while subject 3 learns.

The calculations for \textit{Unknown}, \textit{Known}, and \textit{Coarse} follow the procedure described above. Note that the maximum-information and minimum-information calculations may differ even when there is no equilibrium multiplicity, because subjects may report different WTP values across situations in which only students who are not relevant to them change their learning decisions.

\newpage
\section{Supplementary Tables and Figures}\label{sec:supp_figures}

\begin{table}[H]
\centering
\caption{Comparison of Total Student Payoff in Experiment 2}
\label{tab:welfare_exp2_costs}
\resizebox{\textwidth}{!}{
\begin{tabular}{ccc ccc cc cc cc}
\toprule
\multicolumn{3}{c}{Costs} & \multicolumn{3}{c}{Total payoff} & \multicolumn{2}{c}{$Known-Coarse$} & \multicolumn{2}{c}{$Known-Unknown$} & \multicolumn{2}{c}{$Coarse-Unknown$} \\
\cmidrule(lr){1-3} \cmidrule(lr){4-6} \cmidrule(lr){7-8} \cmidrule(lr){9-10} \cmidrule(lr){11-12}
S1 & S2 & S3 & \textit{Unknown} & \textit{Known} & \textit{Coarse} & Diff. & $p$-value & Diff. & $p$-value & Diff. & $p$-value \\
\midrule
1.8&1.8&1.8&82.88&84.77&83.53&1.23&0.233&\textcolor{Maroon}{1.89}&\textcolor{Maroon}{0.035}&0.65&0.443\\
1.8&1.8&2.6&81.85&84.33&83.53&0.80&0.433&\textcolor{Maroon}{2.48}&\textcolor{Maroon}{0.003}&\textcolor{Maroon}{1.68}&\textcolor{Maroon}{0.022}\\
1.8&1.8&3.0&81.57&84.78&83.53&1.25&0.206&\textcolor{Maroon}{3.20}&\textcolor{Maroon}{0.000}&\textcolor{Maroon}{1.95}&\textcolor{Maroon}{0.006}\\
1.8&2.6&1.8&81.85&84.03&81.87&\textcolor{Maroon}{2.16}&\textcolor{Maroon}{0.007}&\textcolor{Maroon}{2.18}&\textcolor{Maroon}{0.004}&0.02&0.964\\
1.8&2.6&2.6&81.03&84.28&81.97&\textcolor{Maroon}{2.31}&\textcolor{Maroon}{0.002}&\textcolor{Maroon}{3.25}&\textcolor{Maroon}{0.000}&\textcolor{Maroon}{0.94}&\textcolor{Maroon}{0.043}\\
1.8&2.6&3.0&80.84&83.95&82.06&\textcolor{Maroon}{1.89}&\textcolor{Maroon}{0.015}&\textcolor{Maroon}{3.11}&\textcolor{Maroon}{0.000}&\textcolor{Maroon}{1.22}&\textcolor{Maroon}{0.009}\\
1.8&3.0&1.8&81.57&84.01&81.60&\textcolor{Maroon}{2.41}&\textcolor{Maroon}{0.001}&\textcolor{Maroon}{2.43}&\textcolor{Maroon}{0.000}&0.02&0.960\\
1.8&3.0&2.6&80.84&83.67&81.62&\textcolor{Maroon}{2.05}&\textcolor{Maroon}{0.005}&\textcolor{Maroon}{2.83}&\textcolor{Maroon}{0.000}&\textcolor{Maroon}{0.79}&\textcolor{Maroon}{0.035}\\
1.8&3.0&3.0&80.60&83.62&81.86&\textcolor{Maroon}{1.75}&\textcolor{Maroon}{0.020}&\textcolor{Maroon}{3.01}&\textcolor{Maroon}{0.000}&\textcolor{Maroon}{1.26}&\textcolor{Maroon}{0.001}\\
2.6&1.8&1.8&81.85&82.91&81.87&1.04&0.223&1.06&0.190&0.02&0.964\\
2.6&1.8&2.6&81.03&82.87&81.97&0.91&0.293&\textcolor{Maroon}{1.85}&\textcolor{Maroon}{0.018}&\textcolor{Maroon}{0.94}&\textcolor{Maroon}{0.043}\\
2.6&1.8&3.0&80.84&82.50&82.06&0.44&0.594&\textcolor{Maroon}{1.66}&\textcolor{Maroon}{0.022}&\textcolor{Maroon}{1.22}&\textcolor{Maroon}{0.009}\\
2.6&2.6&1.8&81.03&82.55&80.81&\textcolor{Maroon}{1.74}&\textcolor{Maroon}{0.013}&\textcolor{Maroon}{1.53}&\textcolor{Maroon}{0.019}&-0.21&0.579\\
2.6&2.6&2.6&80.39&82.54&81.02&1.51&0.052&\textcolor{Maroon}{2.15}&\textcolor{Maroon}{0.002}&0.64&0.129\\
2.6&2.6&3.0&80.12&82.60&81.03&\textcolor{Maroon}{1.57}&\textcolor{Maroon}{0.023}&\textcolor{Maroon}{2.47}&\textcolor{Maroon}{0.000}&\textcolor{Maroon}{0.91}&\textcolor{Maroon}{0.009}\\
2.6&3.0&1.8&80.84&82.09&80.42&\textcolor{Maroon}{1.68}&\textcolor{Maroon}{0.014}&\textcolor{Maroon}{1.26}&\textcolor{Maroon}{0.049}&-0.42&0.204\\
2.6&3.0&2.6&80.12&82.32&80.71&\textcolor{Maroon}{1.61}&\textcolor{Maroon}{0.015}&\textcolor{Maroon}{2.20}&\textcolor{Maroon}{0.000}&\textcolor{Maroon}{0.59}&\textcolor{Maroon}{0.016}\\
2.6&3.0&3.0&80.00&82.30&80.87&\textcolor{Maroon}{1.43}&\textcolor{Maroon}{0.033}&\textcolor{Maroon}{2.30}&\textcolor{Maroon}{0.000}&\textcolor{Maroon}{0.87}&\textcolor{Maroon}{0.002}\\
3.0&1.8&1.8&81.57&81.85&81.60&0.25&0.758&0.28&0.721&0.02&0.960\\
3.0&1.8&2.6&80.84&81.80&81.62&0.17&0.819&0.96&0.168&\textcolor{Maroon}{0.79}&\textcolor{Maroon}{0.035}\\
3.0&1.8&3.0&80.60&81.75&81.86&-0.11&0.896&1.15&0.127&\textcolor{Maroon}{1.26}&\textcolor{Maroon}{0.001}\\
3.0&2.6&1.8&80.84&81.67&80.42&1.25&0.058&0.83&0.178&-0.42&0.204\\
3.0&2.6&2.6&80.12&81.43&80.71&0.72&0.252&\textcolor{Maroon}{1.30}&\textcolor{Maroon}{0.027}&\textcolor{Maroon}{0.59}&\textcolor{Maroon}{0.016}\\
3.0&2.6&3.0&80.00&81.44&80.87&0.57&0.443&\textcolor{Maroon}{1.44}&\textcolor{Maroon}{0.038}&\textcolor{Maroon}{0.87}&\textcolor{Maroon}{0.002}\\
3.0&3.0&1.8&80.60&81.72&80.49&\textcolor{Maroon}{1.23}&\textcolor{Maroon}{0.049}&1.12&0.055&-0.11&0.690\\
3.0&3.0&2.6&80.00&81.66&80.36&\textcolor{Maroon}{1.31}&\textcolor{Maroon}{0.036}&\textcolor{Maroon}{1.66}&\textcolor{Maroon}{0.004}&0.36&0.159\\
3.0&3.0&3.0&79.87&81.39&80.79&0.61&0.397&\textcolor{Maroon}{1.52}&\textcolor{Maroon}{0.021}&\textcolor{Maroon}{0.92}&\textcolor{Maroon}{0.001}\\
\bottomrule
\end{tabular}
}
\begin{tabnotes}
The table reports total student payoff by treatment for each cost combination in Experiment 2. The first three columns report the information costs for students with priorities 1, 2, and 3, respectively. The columns labeled ``Diff.'' report pairwise differences in average total student payoff between treatments, with the corresponding $p$-values reported in the adjacent columns. Standard errors are adjusted using the method of \citet{abrevaya2008recombinant}. Entries highlighted in \textcolor{Maroon}{Maroon} indicate differences that are statistically significant at the 5\% level.
\end{tabnotes}
\end{table}

\begin{table}[htbp]
\centering
\caption{Comparison of Total Student Payoff in Experiment 1 (Experiment-2 Costs)}
\label{tab:welfare_exp1_exp2cost_max}
\resizebox{\textwidth}{!}{
\begin{tabular}{ccc ccc cc cc cc}
\toprule
\multicolumn{3}{c}{Costs} & \multicolumn{3}{c}{Total payoff} & \multicolumn{2}{c}{$Known-Coarse$} & \multicolumn{2}{c}{$Known-Unknown$} & \multicolumn{2}{c}{$Coarse-Unknown$} \\
\cmidrule(lr){1-3} \cmidrule(lr){4-6} \cmidrule(lr){7-8} \cmidrule(lr){9-10} \cmidrule(lr){11-12}
S1 & S2 & S3 & \textit{Unknown} & \textit{Known} & \textit{Coarse} & Diff. & $p$-value & Diff. & $p$-value & Diff. & $p$-value \\
\midrule
1.8&1.8&1.8&82.63&85.41&84.10&1.31&0.208&\textcolor{Maroon}{2.77}&\textcolor{Maroon}{0.002}&1.47&0.081\\
1.8&1.8&2.6&81.85&85.29&84.24&1.05&0.307&\textcolor{Maroon}{3.44}&\textcolor{Maroon}{0.000}&\textcolor{Maroon}{2.39}&\textcolor{Maroon}{0.002}\\
1.8&1.8&3.0&81.79&85.24&84.24&0.99&0.336&\textcolor{Maroon}{3.44}&\textcolor{Maroon}{0.000}&\textcolor{Maroon}{2.45}&\textcolor{Maroon}{0.001}\\
1.8&2.6&1.8&81.85&84.61&82.44&\textcolor{Maroon}{2.17}&\textcolor{Maroon}{0.007}&\textcolor{Maroon}{2.75}&\textcolor{Maroon}{0.000}&0.59&0.327\\
1.8&2.6&2.6&81.14&84.49&82.61&\textcolor{Maroon}{1.89}&\textcolor{Maroon}{0.019}&\textcolor{Maroon}{3.35}&\textcolor{Maroon}{0.000}&\textcolor{Maroon}{1.47}&\textcolor{Maroon}{0.006}\\
1.8&2.6&3.0&80.83&84.43&82.61&\textcolor{Maroon}{1.83}&\textcolor{Maroon}{0.025}&\textcolor{Maroon}{3.60}&\textcolor{Maroon}{0.000}&\textcolor{Maroon}{1.77}&\textcolor{Maroon}{0.000}\\
1.8&3.0&1.8&81.79&84.43&82.02&\textcolor{Maroon}{2.41}&\textcolor{Maroon}{0.001}&\textcolor{Maroon}{2.63}&\textcolor{Maroon}{0.000}&0.22&0.660\\
1.8&3.0&2.6&80.83&84.31&82.19&\textcolor{Maroon}{2.12}&\textcolor{Maroon}{0.005}&\textcolor{Maroon}{3.48}&\textcolor{Maroon}{0.000}&\textcolor{Maroon}{1.36}&\textcolor{Maroon}{0.002}\\
1.8&3.0&3.0&80.58&84.25&82.19&\textcolor{Maroon}{2.06}&\textcolor{Maroon}{0.007}&\textcolor{Maroon}{3.68}&\textcolor{Maroon}{0.000}&\textcolor{Maroon}{1.61}&\textcolor{Maroon}{0.000}\\
2.6&1.8&1.8&81.85&83.26&82.44&0.82&0.401&1.41&0.122&0.59&0.327\\
2.6&1.8&2.6&81.14&83.14&82.61&0.54&0.580&\textcolor{Maroon}{2.01}&\textcolor{Maroon}{0.024}&\textcolor{Maroon}{1.47}&\textcolor{Maroon}{0.006}\\
2.6&1.8&3.0&80.83&83.09&82.61&0.48&0.623&\textcolor{Maroon}{2.25}&\textcolor{Maroon}{0.010}&\textcolor{Maroon}{1.77}&\textcolor{Maroon}{0.000}\\
2.6&2.6&1.8&81.14&82.70&81.11&1.59&0.051&\textcolor{Maroon}{1.57}&\textcolor{Maroon}{0.036}&-0.02&0.959\\
2.6&2.6&2.6&80.48&82.59&81.29&1.30&0.112&\textcolor{Maroon}{2.11}&\textcolor{Maroon}{0.004}&0.82&0.055\\
2.6&2.6&3.0&80.29&82.53&81.29&1.24&0.132&\textcolor{Maroon}{2.25}&\textcolor{Maroon}{0.002}&\textcolor{Maroon}{1.01}&\textcolor{Maroon}{0.012}\\
2.6&3.0&1.8&80.83&82.56&80.79&\textcolor{Maroon}{1.77}&\textcolor{Maroon}{0.021}&\textcolor{Maroon}{1.73}&\textcolor{Maroon}{0.015}&-0.04&0.914\\
2.6&3.0&2.6&80.29&82.45&80.98&1.46&0.056&\textcolor{Maroon}{2.16}&\textcolor{Maroon}{0.002}&\textcolor{Maroon}{0.70}&\textcolor{Maroon}{0.034}\\
2.6&3.0&3.0&80.09&82.39&80.98&1.41&0.069&\textcolor{Maroon}{2.30}&\textcolor{Maroon}{0.001}&\textcolor{Maroon}{0.89}&\textcolor{Maroon}{0.005}\\
3.0&1.8&1.8&81.79&82.30&82.02&0.28&0.763&0.50&0.565&0.22&0.660\\
3.0&1.8&2.6&80.83&82.18&82.19&-0.01&0.991&1.35&0.113&\textcolor{Maroon}{1.36}&\textcolor{Maroon}{0.002}\\
3.0&1.8&3.0&80.58&82.13&82.19&-0.07&0.942&1.55&0.069&\textcolor{Maroon}{1.61}&\textcolor{Maroon}{0.000}\\
3.0&2.6&1.8&80.83&81.87&80.79&1.07&0.159&1.03&0.145&-0.04&0.914\\
3.0&2.6&2.6&80.29&81.75&80.98&0.77&0.311&\textcolor{Maroon}{1.47}&\textcolor{Maroon}{0.039}&\textcolor{Maroon}{0.70}&\textcolor{Maroon}{0.034}\\
3.0&2.6&3.0&80.09&81.70&80.98&0.72&0.353&\textcolor{Maroon}{1.61}&\textcolor{Maroon}{0.024}&\textcolor{Maroon}{0.89}&\textcolor{Maroon}{0.005}\\
3.0&3.0&1.8&80.58&81.74&80.49&1.25&0.088&1.16&0.090&-0.08&0.793\\
3.0&3.0&2.6&80.09&81.63&80.68&0.94&0.197&\textcolor{Maroon}{1.54}&\textcolor{Maroon}{0.026}&\textcolor{Maroon}{0.60}&\textcolor{Maroon}{0.027}\\
3.0&3.0&3.0&79.91&81.57&80.68&0.89&0.231&\textcolor{Maroon}{1.66}&\textcolor{Maroon}{0.017}&\textcolor{Maroon}{0.78}&\textcolor{Maroon}{0.003}\\
\bottomrule
\end{tabular}
}
\begin{tabnotes}
The table reports total student payoff by treatment in Experiment 1, under the maximum information equilibrium, for each Experiment-2 cost combination. The first three columns report the information costs for students with priorities 1, 2, and 3, respectively. The columns labeled ``Diff.'' report pairwise differences in average total student payoff between treatments, with the corresponding $p$-values reported in the adjacent columns. Standard errors are adjusted using the method of \citet{abrevaya2008recombinant}. Entries highlighted in \textcolor{Maroon}{Maroon} indicate differences that are statistically significant at the 5\% level.
\end{tabnotes}
\end{table}

\begin{figure}[htbp]
  \centering

  \begin{subfigure}{0.75\textwidth}
    \centering
    \includegraphics[width=\linewidth]{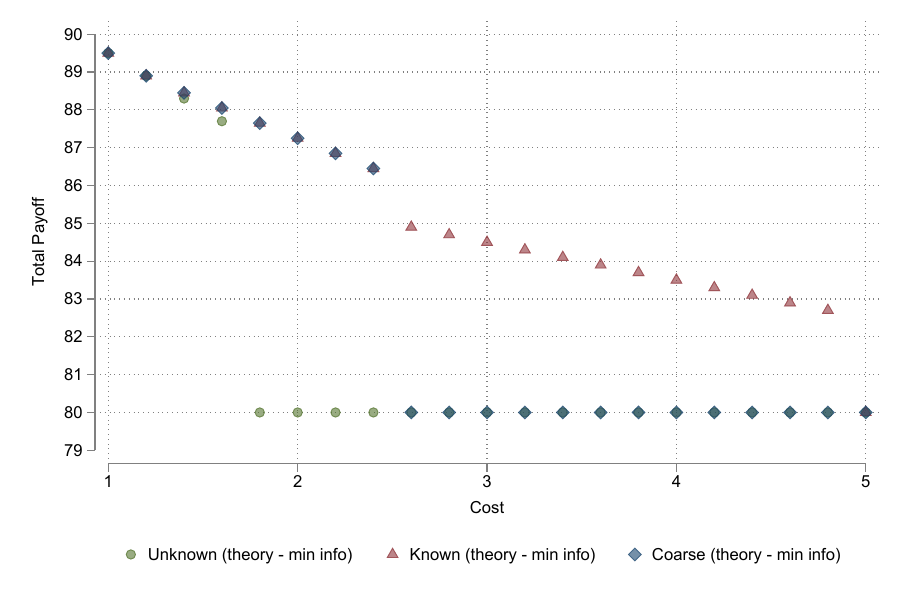}
    \caption{Theory}
    \label{fig:Welfare_theory_sym_min}
  \end{subfigure}

  \vspace{0.5em}

  \begin{subfigure}{0.75\textwidth}
    \centering
    \includegraphics[width=\linewidth]{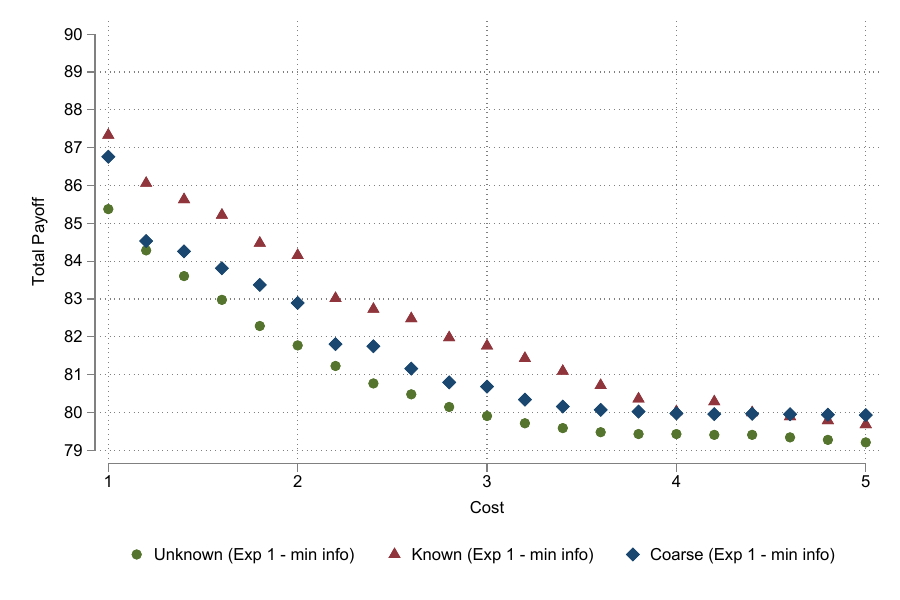}
    \caption{Experiment 1}
    \label{fig:Welfare_exp1_sym_min}
  \end{subfigure}

  \caption{Total Student Payoff: Symmetric Costs, Minimum Information}
  \label{fig:Welfare_sym_min}

  \begin{tabnotes}
      The figures show total student payoff under each treatment for all symmetric costs in $[1,5]$ on a 0.2 grid under the minimum-information equilibrium. Panel (a) presents the theoretical prediction under minimum-information equilibrium selection, and Panel (b) presents the results from Experiment 1, fixing beliefs at the minimum-information equilibrium.
  \end{tabnotes}
\end{figure}

\begin{figure}[htbp]
 \centering
 \includegraphics[width=1\linewidth]{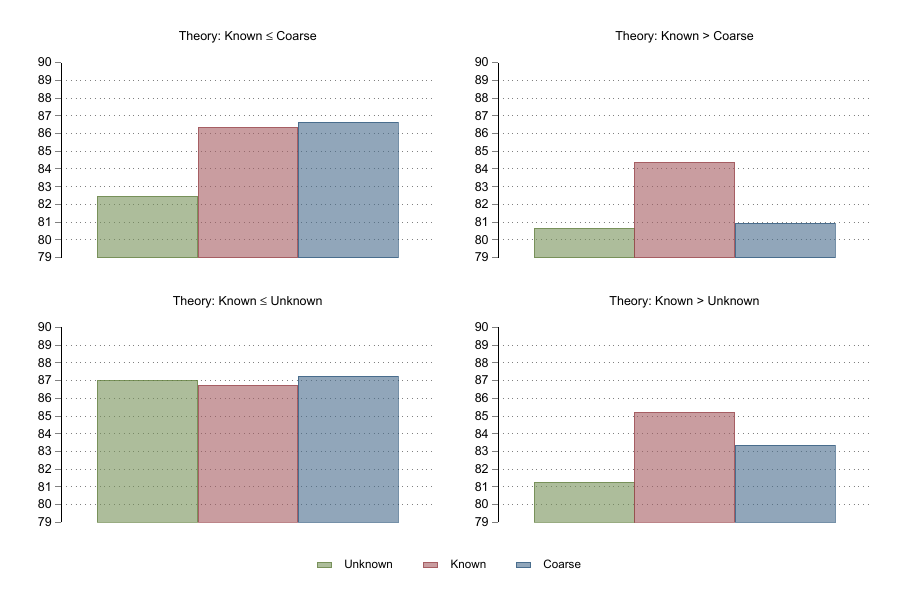}
 \caption{Total Student Payoff: Theory (All Costs)}
 \label{fig:Welfare_theory_all}
 \begin{tabnotes}
     The figure reports theoretically predicted total student payoff for all cost combinations in $[1,5]^3$ on a 0.2 grid under the maximum-information equilibrium. Cost combinations are grouped according to the theoretically predicted treatment ranking.
 \end{tabnotes}
\end{figure}

\begin{figure}[h]
 \centering
 \includegraphics[width=0.8\linewidth]{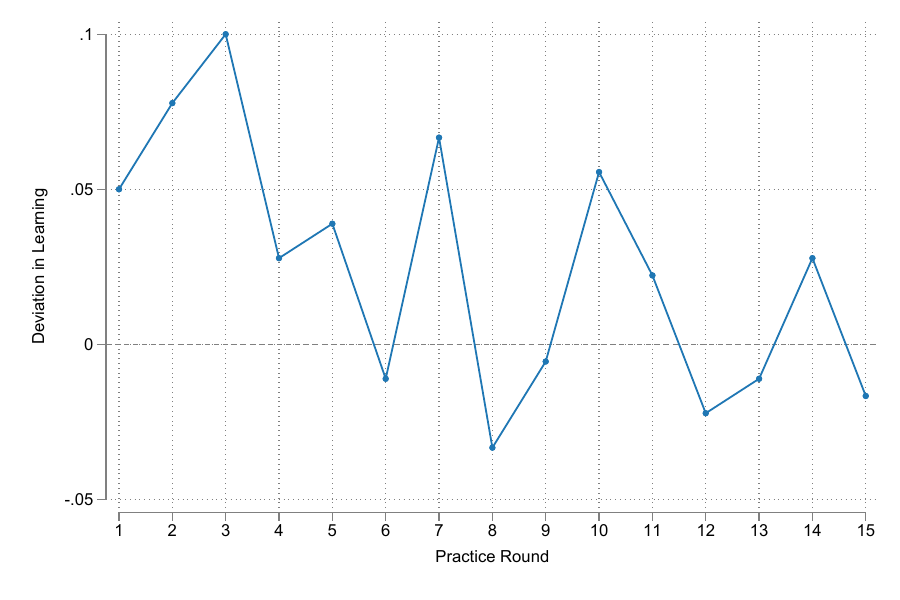}
 \caption{Deviation in Learning Decisions (Practice Rounds)}
 \label{fig:deviation_practice}
    \begin{tabnotes}
        The figure displays the average deviation of subjects' learning decisions from the theoretical predictions during the practice rounds. We code ``Learn'' as 1 and ``Not Learn'' as 0, and define deviation as the experimental decision minus the theoretical prediction. Thus, zero indicates correct learning, positive values indicate over-search on average, and negative values indicate under-search on average.
    \end{tabnotes}
\end{figure}

\begin{table}[h]
\centering
\caption{Determinants of WTP by Treatments}
\def\sym#1{\ifmmode^{#1}\else\(^{#1}\)\fi}
\begin{tabular*}{1\hsize}{@{\hskip\tabcolsep\extracolsep\fill}l*{3}{cc}}
\hline\hline
            &\multicolumn{6}{c}{WTP}\\
\cline{2-7}
            &\multicolumn{2}{c}{(1)}           &\multicolumn{2}{c}{(2)}           &\multicolumn{2}{c}{(3)}           \\
            &\multicolumn{2}{c}{(Unknown)}&\multicolumn{2}{c}{(Known)}&\multicolumn{2}{c}{(Coarse)}\\ 
\hline\hline
%switch      &                     &            &                     &            &                     &            \\
\textit{Student 1}       &                     &            &       2.366\sym{***}&     (0.416)&                     &            \\
\textit{Student 2}        &                     &            &       1.025\sym{***}&     (0.306)&                     &            \\
\textit{Student 1/2}      &                     &            &                     &            &       1.346\sym{***}&     (0.241)\\
\textit{Curiosity}   &       0.123         &     (0.148)&       0.829\sym{***}&     (0.159)&       0.280         &     (0.229)\\
\textit{Female}      &      -0.621\sym{*}         &     (0.336)&      -0.464         &     (0.302)&      -0.132         &     (0.363)\\
\textit{Graduate}        &     -0.0717         &     (0.353)&       1.022         &     (0.642)&     -0.0544         &     (0.465)\\
\textit{Age}          &      0.0309         &    (0.0219)&     -0.0998         &    (0.0725)&      0.0178         &    (0.0314)\\
Constant      &       1.423\sym{**}  &     (0.565)&       2.256         &     (1.396)&    -0.00274         &     (0.646)\\
\hline
\(N\)       &         189         &            &         252         &            &         198         &            \\
\hline\hline
\end{tabular*}
\begin{tabnotes}
    The table reports coefficients from regressions of subjects' WTP for information about School B on subjects' priority indicators, curiosity measured in Experiment 3, and demographic characteristics. \textit{Student 1} and \textit{Student 2} are dummy variables that equal one for the respective priorities in Treatment \textit{Known} and zero otherwise. \textit{Student 1/2} is a dummy that equals one for top two students in Treatment \textit{Coarse} and zero otherwise. \textit{Curiosity} is the WTP for information when the admission chance is zero in Experiment 3. \textit{Female} is a dummy that equals one for female subjects and zero otherwise. \textit{Graduate} is a dummy that equals one for graduate students and zero for undergraduate students. \textit{Age} represents the age of subjects, ranging from 18 to 52 years old. Tobit models are adopted with censoring from below and above. Standard errors are clustered at the level of individual subjects. \sym{*} \(p<0.1\), \sym{**} \(p<0.05\), \sym{***} \(p<0.01\)
\end{tabnotes}
\label{tab:regression WTP_treatment}
\end{table}

\end{document}